\def\version{November 25, 2014}

\documentclass[12pt]{article}
\def\macrosPb{} 
\def\macrosHarxiv{} 
\usepackage[psamsfonts]{amsfonts}
\usepackage{amsmath,amssymb,amsthm}
\usepackage[dvips]{graphicx}
\usepackage{appendix}
\usepackage{bbm} 
\usepackage{amsbsy}
\usepackage{enumerate}
\usepackage{cite}

\ifdefined\macrosPa
  \usepackage[textwidth=465pt,textheight=650pt,centering]{geometry} 
\else\ifdefined\macrosPb
  \usepackage[textwidth=500pt,textheight=650pt,centering]{geometry} 
\fi\fi

\ifdefined\macrosS
  \makeatletter
  
  \def\boldsymbol{\pmb}
  \makeatother

  \usepackage{mathptmx}
  \DeclareMathAlphabet{\mathcal}{OMS}{cmsy}{m}{n}
\fi

\def\UseSection{
        \numberwithin{equation}{section}
	\theoremstyle{plain}
        \newtheorem{theorem}    {Theorem}[section]
        \DefineTheorems 
}

\def\DefineTheorems{
	
	\newtheorem{lemma}      [theorem] {Lemma}
	
	\newtheorem{prop}       [theorem] {Proposition}
	
	\newtheorem{cor}        [theorem] {Corollary}

	\theoremstyle{definition}
	\newtheorem{defn}       [theorem] {Definition}
	
	\newtheorem{example}       [theorem] {Example}
	\newtheorem{rk} 	[theorem] {Remark}
	\theoremstyle{definition}

}

\newcommand{\bt}   {\begin{theorem}}
\newcommand{\et}   {\end  {theorem}}
\newcommand{\bl}   {\begin{lemma}}
\newcommand{\el}   {\end  {lemma}}
\newcommand{\bp}   {\begin{prop}}
\newcommand{\ep}   {\end  {prop}}
\newcommand{\bc}   {\begin{cor}}
\newcommand{\ec}   {\end  {cor}}
\newcommand{\bd}   {\begin{defn}}
\newcommand{\ed}   {\end  {defn}}

\newcommand{\ba}   {\begin{array}}
\newcommand{\ea}   {\end  {array}}
\newcommand{\be}   {\begin{enumerate}}
\newcommand{\ee}   {\end  {enumerate}}
\newcommand{\bi}   {\begin{itemize}}
\newcommand{\ei}   {\end  {itemize}}

\def\eq#1\en{\begin{equation}#1\end{equation}}  
\def\eqsplit#1\ensplit{
	\begin{equation}\begin{split}#1\end{split}\end{equation}
	}
\def\eqalign#1\enalign{
	\begin{align}#1\end{align}
	}
\def\eqmul#1\enmul{
	\begin{multline}#1\end{multline}
	}
\newcommand{\eqarrstar} {\begin{eqnarray*}} 
\newcommand{\enarrstar} {\end{eqnarray*}} 
\newcommand{\eqarray}   {\begin{eqnarray}} 
\newcommand{\enarray}   {\end{eqnarray}} 
\newcommand{\nnb}	{\nonumber \\} 

\newcommand{\lbeq}[1]  {\label{e:#1}}
\newcommand{\refeq}[1] {\eqref{e:#1}}    

\makeatletter
\newcommand{\labelcounter}[2]{{%
	\stepcounter{#1}
	\protected@write\@auxout{}%
	{\string\newlabel{#2}{{\csname the#1\endcsname}{\thepage}}}%
	{\ref{#2}}
	}}
\makeatother
\newcommand{\Cbold} {{\mathbb C}}  
  
\newcommand{\Ebold} {{\mathbb E}}

\newcommand{\Nbold} {{\mathbb N}}

\newcommand{\Rbold} {{\mathbb R}}

\newcommand{\Zbold} {{\mathbb Z}}

\newcommand{\Acal}   {\mathcal{A}} 
\newcommand{\Bcal}   {\mathcal{B}} 
\newcommand{\Ccal}   {\mathcal{C}}

\newcommand{\Fcal}   {\mathcal{F}} 
 
\newcommand{\Hcal}   {\mathcal{H}}

\newcommand{\Ncal}   {\mathcal{N}} 
 
\newcommand{\Pcal}   {\mathcal{P}}

\newcommand{\Rcal}   {\mathcal{R}}
\newcommand{\Scal}   {\mathcal{S}} 
 
\newcommand{\Ucal}   {\mathcal{U}}

\newcommand{\Zd}    {{ {\Zbold}^d }}

\newcommand{\spose}[1] {{\hbox to 0pt{#1\hss}} }
\newcommand{\ltapprox} {\mathrel{\spose{\lower 3pt\hbox{$\mathchar"218$}}
 \raise 2.0pt\hbox{$\mathchar"13C$}}}
\newcommand{\gtapprox} {\mathrel{\spose{\lower 3pt\hbox{$\mathchar"218$}}
 \raise 2.0pt\hbox{$\mathchar"13E$}}}

\UseSection   
\usepackage[usenames]{color}

\definecolor{at}{rgb}{0.0, 0.5, 0.0} 

\renewcommand{\to} {\rightarrow}

\newcommand{\sumtwo}[2]{\sum_{ \mbox{ \scriptsize
    $\begin{array}{c}
                        {#1} \\ {#2}
                        \end{array} $ }
    }
}

\newcommand{\R}{\Rbold}

\newcommand{\N}{\Nbold}
\newcommand{\C}{\mathbb{C}}

\newcommand{\Lambdabold}{\boldsymbol{\Lambda}}

\newcommand{\1}{\mathbbm{1}}
\newcommand{\Cbf}{\boldsymbol{C}}

\newcommand{\Abf}{\boldsymbol{A}}

\newcommand{\psib}{\bar\psi}

\newcommand{\Ex}{\mathbb{E}}

\newcommand{\Econstg}{\alpha_{G}}

\newcommand{\pair}[1]{\langle #1 \rangle}

\newcommand{\Phipoltil}{\widetilde{\Pi}}

\newcommand{\units}{\Ucal}

\newcommand{\concat}{\circ}

\newcommand{\Ttimes}{T}

\newcommand{\h}{\mathfrak{h}}

\newcommand{\poly}{P}

\newcommand{\species}{\mathbf{s}}
\newcommand{\sgn}{\mathrm{sgn}}

\ifdefined\macrosH
  \usepackage{xr-hyper}
  \usepackage{hyperref}
  \hypersetup{hypertexnames=false}
  \hypersetup{colorlinks,citecolor=blue,linkcolor=blue}  

\else\ifdefined\macrosHarxiv
  \usepackage{xr-hyper}
  \usepackage{hyperref}
  \hypersetup{hypertexnames=false}

\else
  \newcommand{\texorpdfstring}[2]{#1}
  \usepackage{xr}
\fi\fi

\title  {
       A renormalisation group method.
       \\
       I.
       Gaussian integration and normed algebras
        }

\author{
David C. Brydges\thanks{Department of Mathematics,
University of British Columbia,
Vancouver, BC, Canada V6T 1Z2.
E-mail: {\tt db5d@math.ubc.ca}, {\tt slade@math.ubc.ca}.}\;
 and Gordon Slade$^*$
}

\date\version

\begin{document}

\maketitle

\begin{abstract}
This paper is the first in a series devoted to the development of a
rigorous renormalisation group method for lattice field theories
involving boson fields, fermion fields, or both.  Our immediate
motivation is a specific model, involving both boson and fermion
fields, which arises as a representation of the continuous-time weakly
self-avoiding walk.  In this paper, we define normed algebras suitable
for a renormalisation group analysis, and develop methods for
performing analysis on these algebras.  We also develop the theory of
Gaussian integration on these normed algebras, and prove estimates for
Gaussian integrals.  The concepts and results developed here provide a
foundation for the continuation of the method presented in subsequent
papers in the series.
\end{abstract}

\section{Introduction}
\label{sec:intro}

This paper is the first in a series devoted to the development of a
rigorous renormalisation group method.  We develop the method with the
specific goal of providing the necessary ingredients for our analysis
of the critical behaviour of the continuous-time weakly self-avoiding
walk in dimension~4 \cite{BBS-saw4-log,BBS-saw4}, via its
representation as a supersymmetric field theory involving both boson
and fermion fields \cite{BIS09}.  However, our approach is more
general, and also applies in other settings, including purely bosonic
or purely fermionic field theories.  In particular, it is applied to
the 4-dimensional $n$-component $|\varphi|^4$ model in
\cite{BBS-phi4-log}.  Other approaches to the rigorous renormalisation
group are discussed in \cite{Bryd09}.

In the renormalisation group approach, we are interested in performing
a Gaussian integral with respect to a positive-definite covariance
operator $C$.  The integration is performed progressively: the
covariance is decomposed as a sum of positive-definite terms
$C=C_1+C'$ and the original integral is equal to a convolution of
Gaussian integrals with respect to $C_1$ and $C'$.  A proof that
decomposition of the covariance corresponds to convolution of Gaussian
integrals can be found for our context in \cite{BI03d}, but we will
give a self-contained proof here within our current formalism and
notation.

In order to perform analysis with Gaussian integrals, it is necessary
to define suitable norms.  In this paper, we define an algebra $\Ncal$
and the $T_\phi$ \emph{semi-norm} on $\Ncal$, and prove that the
$T_\phi$ semi-norm obeys an essential product property.
We prove several estimates for the $T_\phi$ semi-norm,
which are essential for our renormalisation group method,
including estimates for Gaussian integrals.  In addition, as an
example of use of the $T_\phi$ semi-norm, and as preparation for more
detailed estimates obtained in \cite{BS-rg-IE}, we prove a preliminary
estimate for the self-avoiding walk interaction.

The concepts
and results from this paper that are needed in subsequent papers in
the series are summarised in Section~\ref{sec:gint}, which pertains
to Gaussian integration, and in Section~\ref{sec:Tphi-props}, which
pertains to norms and norm estimates.
Most of the proofs are deferred to later sections.

\section{Gaussian integration}
\label{sec:gint}

\subsection{Fields and the algebra \texorpdfstring{$\Ncal$}{Ncal}}
\label{sec:Ncal}

Given a finite set $\Lambdabold$, and $p\in\N$, let $\Lambdabold^p$
denote the $p$-fold cartesian product of $\Lambdabold$ with itself, so
that elements of $\Lambdabold^p$ are sequences of elements of
$\Lambdabold$ of \emph{length} $p$.  We define
$\Lambdabold^0=\{\varnothing\}$ to be the set whose element is the
empty sequence.  Then $\Lambdabold^{*} = \sqcup_{p=0}^\infty
\Lambdabold^p$ is the set of arbitrary finite sequences of elements of
$\Lambdabold$, of any length, including zero.  We typically denote the
length of $z \in \Lambdabold^{*}$ as $p=p (z)$ or $q=q(z)$, and, for
$z \in \Lambdabold^{*}$, we write $z! = p(z)!$.  For $z', z'' \in
\Lambdabold^{*}$ we define the \emph{concatenation} $z'\concat z''$ to
be the sequence in $\Lambdabold^{*}$ whose elements are the elements
of $z'$ followed by the elements of $z''$.

Let $\Lambdabold_b$ be any finite set.  An element of
$\R^{\Lambdabold_b}$ is called a \emph{boson field}, and can be
written as $\phi = (\phi_{x},\;x \in \Lambdabold_{b})$.  Let
$\Rcal=\Rcal(\Lambdabold_b)$ denote the ring of smooth functions from
$\R^{\Lambdabold_b}$ to $\C$.  Here \emph{smooth} means having at
least $p_\Ncal$ continuous derivatives, where $p_\Ncal$ is a parameter
at our disposal.

Let $\Lambdabold_f$ be a finite set and let $\Lambdabold =
\Lambdabold_b \sqcup \Lambdabold_f$.  The \emph{fermion field} $\psi =
(\psi_{y}, y \in \Lambdabold_{f})$ is a set of anticommuting
generators for an algebra $\Ncal=\Ncal (\Lambdabold)$ over the ring
$\Rcal$.  In particular, $\psi_y^2=0$ for all $y \in \Lambdabold_f$.
By definition, $\Ncal$ consists of elements $F$ of the form
\begin{equation}
    \label{e:K}
    F
=
    \sum_{y \in \Lambdabold_f^*} \frac{1}{y!} F_y \psi^y
,
\end{equation}
where each coefficient $F_{y}$ is an element of $\Rcal$, and
\begin{equation}
\lbeq{psiy}
    \psi^y = \begin{cases}
    1 & \text{if $q(y)=0$}
    \\
    \psi_{y_1}\cdots \psi_{y_q} & \text{if $q \geq 1$ and $y=(y_1,\ldots,y_q)$}.
    \end{cases}
\end{equation}
We always require $F_{y}$ to be antisymmetric under permutation of the
components of $y$; this ensures that the representation \refeq{K} is
unique.  We denote $F_{y}$ evaluated at $\phi$ by $F_{y} (\phi)$, and
write $F(\phi)=\sum_{y \in \Lambdabold_f^*} \frac{1}{y!} F_y(\phi)
\psi^y$.  Given $x\in \Lambdabold_b^*$, we define $\phi^x$ in the same
way as \refeq{psiy}.

\begin{defn}
\label{def:Npoly} For $A$ a nonnegative integer, we say that $F \in
\Ncal$ is a \emph{polynomial of degree} $A$ if there are coefficients
$F_{x,y}\in\C$ such that $F(\phi) = \sum_{x,y: p(x)+q(y) \le A}
\frac{1}{x!y!} F_{x,y} \phi^x \psi^y$, with $F_{x,y}\neq 0$ for some
$x,y$ with $p(x)+q(y)=A$.
\end{defn}

Polynomial elements of $\Ncal$ play an important role in our analysis.
An example of a polynomial of degree 2 is $\phi_w\phi_x +
\psi_y\psi_z$, for some $w,x\in \Lambdabold_b$ and
$y,z\in\Lambdabold_f$.

\subsection{Fermionic Gaussian integration}

Let $\Lambdabold'_{b}$ and $\Lambdabold'_{f}$ be sets, with an order
specified on the elements of $\Lambdabold'_f$.  We
integrate over fields labelled by elements of these sets, starting
in this section
with the fermion fields labelled by $\Lambdabold'_{f}$, and then
in Section~\ref{sec:bgi} with
the boson fields with labels in $\Lambdabold'_{b}$.

We define the monomial $\psi^{\Lambdabold'_{f}}$
to be the product of the generators in the specified order.  Let
$\Lambdabold' = \Lambdabold'_{b}\sqcup \Lambdabold'_{f}$.  We write $F
\in \Ncal (\Lambdabold\sqcup \Lambdabold'_b)$ for the algebra $\Ncal$
with fermion fields indexed by $\Lambdabold_f$ and boson fields
indexed by $\Lambdabold_b \sqcup \Lambdabold_b'$, and $F \in \Ncal
(\Lambdabold\sqcup \Lambdabold')$ for the algebra $\Ncal$ with fermion
fields indexed by $\Lambdabold_f\sqcup \Lambdabold_f'$ and boson
fields indexed by $\Lambdabold_b \sqcup \Lambdabold_b'$.

\begin{defn}
\label{def:grassman-integration}
The \emph{Grassmann integral} is the
linear map $\int_{\Lambdabold'_{f}}:\Ncal
(\Lambdabold\sqcup\Lambdabold')\rightarrow \Ncal
(\Lambdabold\sqcup\Lambdabold'_{b})$ uniquely defined
by the conditions:
\\
(a) for all $F \in \Ncal (\Lambdabold\sqcup
\Lambdabold'_b)$, $\int_{\Lambdabold'_{f}} F\psi^{y'} = 0$
whenever the elements of $y'\in (\Lambdabold_f')^*$ do not
form an enumeration of $\Lambdabold_f'$,
and
\\
(b) $\int_{\Lambdabold'_{f}} F\psi^{\Lambdabold'_{f}} =
F$ for all $F \in \Ncal (\Lambdabold\sqcup
\Lambdabold'_b)$.
\end{defn}

The classic reference for Grassmann integration is
\cite{Bere66}; accessible and more modern treatments can be found in
\cite{CSS13,FKT02,Salm99}.

Given an antisymmetric invertible $\Lambdabold'_{f} \times
\Lambdabold'_{f}$ matrix $\Abf_{f}$, let
\begin{equation}
\label{e:Sfdef}
    S_{f}
=
    \frac{1}{2}\sum_{u,v \in \Lambdabold'_{f}}
    \Abf_{f;u,v} \psi_{u} \psi_{v}
.
\end{equation}
Since the generators anti-commute and since $\Lambdabold'_{f}$ is
finite, the series $\sum_{n=0}^\infty \frac{1}{n!} (-S_{f})^{n}$
terminates after finitely many terms, and therefore defines an element
of $\Ncal (\Lambdabold')$, and hence also of $\Ncal
(\Lambdabold\sqcup\Lambdabold')$.  We denote this element by
$e^{-S_{f}}$.  Let $\Cbf_{f}$ be the inverse of $\Abf_{f}$.  The
Grassmann analogue of Gaussian integration is the linear map
$\Ex_{\Cbf_f}: \Ncal (\Lambdabold\sqcup\Lambdabold') \to \Ncal
(\Lambdabold\sqcup\Lambdabold'_{b})$ defined by
\begin{equation}
    \label{e:ECf}
    \Ex_{\Cbf_f} F = N_{f}\int_{\Lambdabold'_{f}}e^{-S_{f}}F
    ,
    \quad\quad
     F \in \Ncal (\Lambdabold\sqcup\Lambdabold' )
,
\end{equation}
where $N_{f}$ is a normalisation constant such that $ \Ex_{C_f}
1=1$. It is a consequence of \cite[(3.16)]{Bere66} that
\begin{equation}
\label{e:Pfaff}
    N_{f} =
    (\det \Cbf_{f})^{1/2}
.
\end{equation}
The choice of square root depends on the order we have chosen
for $\Lambda'_{f}$.  We will be specific below in a less general
setting.

\subsection{Bosonic Gaussian integration}
\label{sec:bgi}

Given a real symmetric positive-definite $\Lambdabold'_{b}
\times \Lambdabold'_{b}$ matrix $\Abf_b$, and given $\phi \in
\R^{\Lambdabold_b'}$, let
\begin{equation}
\label{e:Sbdef}
    S_{b}
=
    \frac{1}{2}\sum_{u,v \in \Lambdabold'_{b}}
    \Abf_{b;u,v} \phi_{u} \phi_{v}
.
\end{equation}
The matrix $\Abf_b$ has positive eigenvalues and therefore an
inverse matrix $\Cbf_b$ exists.
The Gaussian expectation $\Ex_{\Cbf_b} :
\Ncal(\Lambdabold\sqcup\Lambdabold'_{b}) \to \Ncal (\Lambdabold)$ is
the linear map defined as follows.  Let $D\phi$ be Lebesgue measure on
$\R^{\Lambdabold_b'}$.  For $F \in
\Rcal(\Lambdabold_b\sqcup\Lambdabold'_b)$, we define
\begin{equation}
    \label{e:ECb}
    \Ex_{\Cbf_b} F =   N_{b}\int_{\R^{\Lambdabold'_{b}}} e^{-S_{b}}F
    \,D\phi ,
\end{equation}
where $N_{b}$ is chosen so that $\Ex_{C_b} 1=1$.  It is a standard
fact about Gaussian integrals that $N_{b}$ is given by the
positive square root
\begin{equation}
\label{e:Nb}
    N_b = \left( \det (2\pi \Cbf_b) \right)^{-1/2}.
\end{equation}
Of course $\Ex_{\Cbf_b}$ is only defined on elements of $\Ncal
(\Lambdabold\sqcup\Lambdabold'_{b})$ which are such that the growth of
the coefficients at infinity is not too rapid.  For $F=\sum_{y \in
\Lambdabold_f} \frac{1}{y!}F_y \psi^y \in
\Ncal(\Lambdabold\sqcup\Lambdabold'_{b})$, we define
\begin{equation}
    \Ex_{\Cbf_b} F = \sum_{y \in \Lambdabold_f}\frac{1}{y!}
    (\Ex_{\Cbf_b} F_y) \psi^y.
\end{equation}

\subsection{Combined bosonic-fermionic Gaussian integration on \texorpdfstring{$\Ncal$}{Ncal}}
\label{sec:Grass}

Let $\Cbf$ denote the pair $\Cbf_b,\Cbf_f$.
Given matrices
$\Abf_{f}$ and $\Abf_{b}$ as above, we define the
\emph{combined bosonic-fermionic expectation} to be
the linear map $\Ex_{\Cbf} : \Ncal (\Lambdabold\sqcup \Lambdabold') \to \Ncal
(\Lambdabold)$ given by
\begin{equation}
\label{e:Ecomb}
    \Ex_{\Cbf} = \Ex_{\Cbf_{b}} \Ex_{\Cbf_{f}},
\end{equation}
where $\Ex_{\Cbf_{b}}$ acts only on bosons, and $\Ex_{\Cbf_f}$ acts
only on fermions.
By linearity,
the action of $\Ex_{\Cbf}$ is determined by its action on $KF$ where
$K \in \Ncal(\Lambdabold \sqcup \Lambdabold_b')$ and $F$ is a monomial
in the generators indexed by $\Lambdabold_f'$.
The map $\Ex_{\Cbf}$ is defined, for such $K,F$, by
\begin{align}
    \label{e:ECbf}
    \Ex_{\Cbf} KF
    = ( \Ex_{\Cbf_{b}}K)( \Ex_{\Cbf_{f}}F)
    &=
    \Big(N_b
    \int_{\R^{\Lambdabold'_{b}}}
    e^{-S_b }K \,D\phi\Big)
    \Big(N_{f}\int_{\Lambdabold'_{f}}
    e^{-S_{f}}F\Big)
    .
\end{align}
On the right-hand side, the boson and fermion fields corresponding
to $\Lambdabold'$ have been integrated out, leaving dependence only
on the fields corresponding to $\Lambdabold$.

\subsection{The Laplacian}

It is ordinary calculus to differentiate a function
$f\in \Rcal(\Lambdabold_b)$ with
respect to the components $\phi_u$ of the boson field,
for $u \in \Lambdabold_b$.
The following definition extends this calculus by
providing the standard Grassmann analogue of
differentiation with respect to the fermion field
(see, e.g., \cite{Bere66,FKT02,Salm99}).

\begin{defn}
\label{def:iu}
For $u \in \Lambdabold_f$, the linear map $i_{u} : \Ncal (\Lambdabold)
\rightarrow \Ncal (\Lambdabold )$ is defined
uniquely by the conditions:
\\
(a) $i_{u} (f \psi^y)= f i_u \psi^y$ for $ f \in \Rcal(\Lambdabold_b)$,
$y \in \Lambdabold_f^*$,
\\
(b) $i_u$ acts as an anti-derivation on products of factors of $\psi$,
namely $i_u(\psi^{y_1}\psi^{y_2}) = (i_u\psi^{y_1})\psi^{y_2}
+ (-1)^{p_1}\psi^{y_1}(i_u\psi^{y_2})$, for $y_1,y_2\in \Lambdabold_f^*$
and $p_1$ the length of $y_1$,
and
\\
(c) $i_u \psi_v = \delta_{u,v}$ for $u,v\in\Lambdabold$, where the
right-hand side is the Kronecker delta.
\\
It is natural, and also standard, to write
\eq
    i_u = \frac{\partial}{\partial \psi_{u}}.
\en
By (b) and (c), these operators anti-commute with each
other: $i_ui_v = -i_vi_u$.
\end{defn}

Suppose that there is a bijection $x \mapsto x' =x'(x)$ between a
subset of $\Lambdabold$ and $\Lambdabold '$.  The elements of
$\Lambdabold$ where the bijection is not defined are called
\emph{external}; they do not participate in any integrations.  We
extend the matrices $\Cbf_b,\Cbf_f$ to $\Lambdabold_b \times
\Lambdabold_b$ and $\Lambdabold_f \times \Lambdabold_f$, respectively,
by setting $\Cbf_{b,u',v'}= \Cbf_{f,u',v'}=0$ when $u'$ or $v'$ is
undefined.  We write $\Cbf$ for the pair $\Cbf_b,\Cbf_f$.  The
Laplacian operator $\Delta_{\Cbf} : \Ncal (\Lambdabold ) \rightarrow
\Ncal (\Lambdabold )$ is then defined by
\begin{equation}
\label{e:LapC}
    \Delta_{\Cbf}
=
    \sum_{u,v \in \Lambdabold_{b}}
    \Cbf_{b;u,v}
    \frac{\partial}{\partial \phi_{u}}
    \frac{\partial}{\partial \phi_{v}}
    +
    \sum_{u,v \in \Lambdabold_{f}}
    \Cbf_{f;u,v}
    \frac{\partial}{\partial \psi_{u}}
    \frac{\partial}{\partial \psi_{v}}
    ,
\end{equation}
where the first term on the right-hand side acts only on the
coefficients $F_y(\phi)$ of $F \in \Ncal$, while the second acts
only on the fermionic part $\psi^y$.

\subsection{Gaussian integration and the convolution property}
\label{sec:gi-def}

\begin{example}
\label{ex:conv}
For a bounded function $f$ defined on $\Rbold$ and a probability
measure $\mu$ on $\R$, we can define the convolution $\mu \star f (x)
= \int f (x+y)d\mu (y)$. The map $f \mapsto \mu \star f $ is the
composition of the map $(\theta f) (x,y) = f (x + y)$ followed by
integrating $y$ with respect to $\mu$. The map $\theta$ maps a
function of one variable to a function of two variables.
\end{example}

The following definition implements the construction of
Example~\ref{ex:conv} in the context of the algebra $\Ncal$.  To avoid
simultaneously using $\phi$ to denote a function on $\Lambdabold_{b}$
and a function on the larger space $\Lambdabold_{b} \sqcup
\Lambdabold'_{b}$, we replace $\phi:\Lambdabold_{b} \sqcup
\Lambdabold'_{b}\rightarrow \R$
by
the notation
\begin{equation}\label{e:phi-extended}
    \phi\sqcup\xi:\Lambdabold_{b} \sqcup \Lambdabold'_{b}
\rightarrow
    \R
,
\end{equation}
where $(\phi\sqcup\xi)_{x} = \phi_{x}$ and $(\phi\sqcup\xi)_{x'} =
\xi_{x'}$.
The algebra $\Ncal(\Lambdabold)$ is a subset of $\Ncal(\Lambdabold
\sqcup \Lambdabold')$.

\begin{defn}
\label{def:theta-new} Given $t\in \R$, we define the algebra
homomorphism $\theta_{t} : \Ncal (\Lambdabold) \rightarrow \Ncal
(\Lambdabold \sqcup \Lambdabold')$ to be the unique algebra
homomorphism which obeys:
\\
(a)
the action on generators  $\theta_{t} \psi_{y} = \psi_{y} +
t\psi_{y'}$, for $y \in \Lambdabold_f$, and
\\
(b)
the action on coefficients
$(\theta_{t} f) (\phi,\xi) = f (\phi + t\xi)$, for $f \in \Rcal(\Lambdabold_b)$.
\\
If $x'$ or $y'$ is not defined, as discussed above \refeq{LapC},
then the associated $\xi_{x'}$,
$\psi_{y'}$ is set equal to zero.
Also, on the right-hand side in (b), we interpret $(\phi+t\xi)_x$
as $\phi_x + t\xi_{x'(x)}$.
We write $\theta =\theta_{1}$.
\end{defn}

The following proposition states a convolution property of Gaussian
integrals that is at the heart of the renormalisation group method.

\begin{prop}
\label{prop:conv} For covariances $\Cbf_1,\Cbf_2$ and for
$F \in \Ncal(\Lambdabold)$ such that both sides
of \refeq{conv} are well-defined,
\begin{equation}
\label{e:conv}
    (\Ex_{\Cbf_2} \theta \circ \Ex_{\Cbf_1} \theta ) F
    =
    \Ex_{\Cbf_2+\Cbf_1}\theta F.
\end{equation}
Moreover, if $P \in \Ncal(\Lambdabold)$ is a polynomial of finite
degree, as in Definition~\ref{def:Npoly}, then
\begin{equation}
\label{e:ELap}
    \Ex_{\Cbf} \theta P
    =
    e^{\frac{1}{2} \Delta_{\Cbf}} P.
\end{equation}
\end{prop}

The identity \refeq{conv} follows immediately from \refeq{ELap} for
polynomial $F$, but \refeq{conv} holds more generally.  A proof of
Proposition~\ref{prop:conv} is given in Section~\ref{sec:Gihe}.  The
convolution property \refeq{conv} is standard (see, e.g., \cite{FKT02}
for the purely fermionic version), but our proof follows the approach
in \cite{BI03d} which extends the familiar connection \refeq{ELap}
between Gaussian integration and the Laplacian to the mixed
bosonic-fermionic integral.

The formulas \eqref{e:ELap} and \eqref{e:LapC} compute
moments. For example, if we take $P=\phi_{u} \phi_{v}$ and after
evaluation of \eqref{e:ELap} set $\phi=0$, the result is
$\Ex_{\Cbf} \xi_{u} \xi_{v} = \Cbf_{b;u,v}$. Similarly, by taking
$P=\psi_{u} \psi_{v}$, we obtain $\Ex_{\Cbf} \psi_{u} \psi_{v} = - \Cbf_{f;
u,v}$.  Thus \eqref{e:ELap} is a generalisation of Wick's
theorem (see, e.g., \cite[Lemma~2.3]{BIS09}), which is the standard
formula for moments of a Gaussian measure.

\subsection{Conjugate fermion field}
\label{sec:cff}

Suppose that
$\Lambdabold_f'$ has even cardinality $2M_f$, so
the Grassmann generators can be written in a list as
$\bar\psi_1,\psi_1,\ldots, \bar\psi_{M_f},\psi_{M_f}$, or, more compactly, as
$(\bar\psi_k,\psi_k)_{k=1,\ldots, {M_f}}$.
For the Grassmann generators, there is not a notion of complex
conjugation, so here the bars are used only as a notational
device to list the generators
in pairs.  However, we will still refer to the pairs of generators
as conjugate generators (and see Section~\ref{sec:df} below).
We use the order $\bar\psi_1, \psi_1, \bar\psi_2,\psi_2, \ldots
\bar\psi_{M_f},\psi_{M_f}$ for the generators
in the definition of Grassmann
integration in Definition~\ref{def:grassman-integration}.

Let $A_f$ be an invertible symmetric ${M_f}\times {M_f}$ matrix, with $A_f^{-1}=C_f$.
We define the matrix $\Abf_f$ and its inverse matrix $\Cbf_f$ by
\begin{equation}
\label{e:Afmat}
    \Abf_{f}
=
    \left(
    \begin{array}{cc}
    0&A_f\\
    -A_f^T&0
    \end{array}
    \right)
,
\quad\quad
    \Cbf_{f}
=
    \left(
    \begin{array}{cc}
    0&-C_f^T\\
    C_f&0
    \end{array}
    \right)
,
\end{equation}
with the rows and columns labelled by $\psi_1,\ldots\,
\psi_M,\bar\psi_1,\ldots,
\bar\psi_M$.
Then $S_f$ of \refeq{Sfdef} becomes
\eq
\label{e:Sfsym}
    S_{f}
    =
    \sum_{k,l=1}^{M_f} A_{f;k,l} \psi_k\bar\psi_l
\en
and the normalisation constant $N_f$ of \refeq{Pfaff} is
\eq
\label{e:Pfaffsusy}
    N_f = (\det \Cbf_f)^{1/2} = \det C_f.
\en

For $F$ a monomial in the Grassmann generators, let
$J_F = \Ex_{\Cbf_{f}} F$.
The evaluation of the Grassmann integral $J_F$ is standard
(see, e.g., \cite[Lemma~B.7]{Salm99} or \cite[Proposition~4.1]{BIS09}).
In particular, $J_F=1$ when $F=1$,
$J_F=0$ when
$F=\prod_{r=1}^p\psib_{i_r}\prod_{s=1}^q \psi_{j_s}$ with $p\neq q$,  and
\begin{equation}
\label{e:JF}
    J_F
    =
    \det  C_{f;k_1,\ldots,k_p;l_1,\ldots,l_p},
\end{equation}
when $F=\psib_{k_1}\psi_{l_1}\cdots\psib_{k_p}\psi_{l_p}$,
where $C_{f;k_1,\ldots,k_p;l_1,\ldots,l_p}$
is the $p\times p$ matrix whose $r,s$ element
is $C_{f;k_r,l_s}$.
In particular,
\eq
\label{e:Epsipsi}
    \Ex_{\Cbf_{f}} \bar\psi_k \psi_l = C_{f;k,l},
\en
and $C_f$ is the covariance of the conjugate fermion field.

Conjugate fermion fields will be needed in Proposition~\ref{prop:EK} below.

\subsection{Complex boson field}
\label{sec:cbf}

We now discuss a way to accommodate complex boson fields within the
formalism.  The boson field $\phi$ may include several species of
fields, including real external fields which behave as constants
during integration, and a complex field which does get integrated.  To
describe the latter, we suppose that $\Lambdabold_b'$ has even
cardinality $2M_b$ and write the field as $u_1,v_1,\ldots,
u_{M_b},v_{M_b}$.  Then, for $k=1,\ldots,M_b$, we define
\begin{equation}
    \label{e:phi-def}
    \phi_k = u_k + i v_k,
    \quad
    \bar\phi_k = u_k - i v_k.
\end{equation}
The boson field then corresponds to a complex field
$(\bar\phi_k,\phi_k)_{k=1,\ldots,M_b}$.
Define
\begin{align}
    \label{e:complex-derivs}
    \frac{\partial}{\partial\phi_k}
    &= \frac 12
    \left(
    \frac{\partial}{\partial u_k}
    - i\frac{\partial}{\partial v_k }
    \right),
    \quad\quad
    \frac{\partial}{\partial\bar\phi_k}
    = \frac 12
    \left(
    \frac{\partial}{\partial u_k}
    + i\frac{\partial}{\partial v_k }
    \right).
\end{align}
By definition, these obey, for $k,l=1,\ldots, M_b$,
\begin{align}
    \frac{\partial\phi_k}{\partial\phi_l}
    &= \frac{\partial\bar\phi_k}{\partial\bar\phi_l} = \delta_{k,l},
    \quad\quad
    \frac{\partial\phi_k}{\partial\bar\phi_l}
    =
    \frac{\partial\bar\phi_k}{\partial\phi_{l}}
    =0.
\end{align}

Let $A_b$ be a real invertible symmetric ${M_b}\times {M_b}$ matrix,
with $A_b^{-1}=C_b$.  We define the matrix $\Abf_b$ and its inverse
matrix $\Cbf_b$ by
\begin{equation}
\label{e:Abmat}
    \Abf_{b}
=
    2
    \left(
    \begin{array}{cc}
    A_b&0\\
    0&A_b
    \end{array}
    \right)
,
\quad\quad
    \Cbf_{b}
=
    \frac{1}{2}
    \left(
    \begin{array}{cc}
    C_b&0\\
    0&C_b
    \end{array}
    \right)
,
\end{equation}
with the rows and columns labelled by the real and imaginary parts
$u_1,\ldots,u_{M_b}$, $v_1,\ldots, v_{M_b}$ of the complex boson field.
Then $S_b$ of \refeq{Sbdef} becomes
\begin{equation}
\label{e:Sbsym}
    S_{b}
    =
    \sum_{k,l=1}^{M_b} A_{b;k,l} \phi_k\bar\phi_l
\end{equation}
and the normalisation constant $N_b$ of \refeq{Nb} is
\begin{equation}
    N_b = \left(\det (2\pi \Cbf_b)\right)^{-1/2} = (\det (\pi C_b))^{-1}.
\end{equation}
For $K \in \Ncal(\Lambdabold \sqcup \Lambdabold_b')$, the Gaussian
integral $I_K =\Ex_{\Cbf_b}K$ can equivalently be written as the
complex Gaussian integral
\begin{equation}
\label{e:If}
    I_K = \int_{\C^{M_b}}K d\mu_{C_b}
    \quad\text{with}\quad
    d\mu_{C_b}
    = N_b' \, e^{-S_b}
    \prod_{k=1}^{M_b}\frac{d\bar\phi_k d{\phi}_k}{2\pi i}
,
\end{equation}
where $d\bar\phi_k d{\phi}_k$ is by definition equal to
$2idu_{k}dv_{k}$, where $K$ is considered as a function of
$\bar\phi,\phi$ instead of as a function of the real and imaginary
parts, and where the normalisation constant is
\begin{equation}
    (N_b')^{-1} = \int_{\C^{M_b}}
    e^{-S_b }
    \prod_{k=1}^{M_b}
    \frac{d\bar\phi_k d{\phi}_k}{2\pi i} = \frac{1}{\det  C_b}.
\end{equation}
The factors of $2$ in \eqref{e:Abmat} are included so that
\begin{equation}
\label{e:Ephiphi}
    \Ebold_{\Cbf_b} \bar\phi_k \phi_l = C_{b;k,l},
\end{equation}
and thus we call $C_b$
the covariance of the complex boson field. Expectations of
$\phi$$\phi$ and
$\bar{\phi}$$\bar{\phi}$ are zero.
More generally, expectations of
products of factors of $\phi$ and $\bar\phi$ can be evaluated using
\eqref{e:ELap} together with
\begin{equation}
\label{e:exp-complex-Laplacian}
    \frac 12  \Delta_{\Cbf}
=
    \sum_{k,l=1}^{M_b}
    C_{b;k,l}
    \frac{\partial}{\partial \phi_{k}}
    \frac{\partial}{\partial \bar\phi_{l}}
    +
    \sum_{k,l=1}^{M_f}
    C_{f;k,l}
    \frac{\partial}{\partial \psi_{k}}
    \frac{\partial}{\partial \bar\psi_{l}}
,
\end{equation}
where we computed the Laplacian \eqref{e:LapC} using
\eqref{e:complex-derivs} and \eqref{e:Abmat}.

\subsection{Differential forms}
\label{sec:df}

Suppose we are in the setting of the conjugate fermion field
and complex boson field of
Sections~\ref{sec:cff}--\ref{sec:cbf}, and that $M_f=M_b=M$.
Let
\begin{equation}
\label{e:SA-bis}
    S_{A}
=
    S_{b} + S_{f}
.
\end{equation}
Now \refeq{ECbf} can be written as
\begin{align}
    \label{e:ECbfss}
    \Ex_{\Cbf} K F
    &
    =
    I_KJ_F
    =
    N_b'N_f \int e^{-S_A} KF
    ,
\end{align}
where the Lebesgue measure $D\phi$ has been omitted intentionally
from the right-hand side.  The reason for this omission makes use
of a specific choice of Grassmann generators, as follows.

We choose as Grassmann generators the 1-forms
\begin{align}
    \psi_k &= \frac{1}{(2\pi i)^{1/2}}d\phi_k=
    \frac{1}{(2\pi i)^{1/2}}(du_k + i dv_k),
    \nnb
    \bar\psi_k &= \frac{1}{(2\pi i)^{1/2}}d\bar\phi_k
    = \frac{1}{(2\pi i)^{1/2}}(du_k - i dv_k),
\end{align}
where we fix a choice of square root of $2\pi i$ once and for all.
Multiplication of generators is via the standard anti-commuting wedge
product for differential forms (see, e.g., \cite{Rudi76}); the wedges
are left implicit in what follows.  The 1-forms generate the Grassmann
algebra of differential forms.  In this case the complex
conjugation that acts on the boson field at the same time interchanges
$\psi_k$ and $\bar\psi_k$, but there are no relations other than
anti-commutativity linking the generators of the Grassmann algebra.
Now \refeq{SA-bis} becomes the differential form
\begin{equation}
\label{e:SA-bis-prime}
    S_{A}
=
    \sum_{k,l=1}^M
    \left(
    A_{b;k,l}
    \phi_{k}\bar\phi_{l} +
    \frac{1}{2\pi i}A_{f;k,l}
    d\phi_{k} d\bar\phi_{l}
    \right)
.
\end{equation}
The theory of Gaussian integration in this setting is developed in
\cite{BIS09}.  In particular, it follows from
\cite[Proposition~4.1]{BIS09} that when we interpret the fermionic
part of $e^{-S_f}$ as the differential form $\sum_{n=0}^\infty
\frac{1}{n!}(-S_f)^n$ (the series truncates due to
anti-commutativity), then standard integration of differential forms
gives again
\begin{equation}
\label{e:Efacform}
    \Ex_{\Cbf} K F
    =I_KJ_F.
\end{equation}
Thus Grassmann integral and the standard integration of differential
forms become the same thing.  In the formalism of differential forms,
the omitted Lebesgue measure is supplied by the volume form
$\prod_{k=1}^M d\bar\phi_kd\phi_k$ arising from the expansion of
$e^{-S_f}$. Earlier, we defined $d\bar\phi_k d{\phi}_k$ to be
$2idu_{k}dv_{k}$ because by \eqref{e:phi-def} this is the wedge
product $d\bar\phi_k d{\phi}_k$.

The above shows that
the algebra of differential forms and the form integration used in
\cite{BIS09} is a special case of the construction
of Sections~\ref{sec:cff}--\ref{sec:cbf}.  We do not
need this special case in this paper, but it plays an important
role in \cite{BBS-saw4-log,BBS-saw4}.

\subsection{Supersymmetry}
\label{sec:supersymmetry}

The field theories discussed in \cite{BIS09} and \cite{BBS-saw4-log}
have an additional property of \emph{supersymmetry}: a symmetry
between bosons and fermions.  A discussion of supersymmetry can be
found in \cite[Section~6]{BIS09}.  The field theory becomes
supersymmetric by taking $M_b=M_f=M$ and choosing the boson and
fermion covariances to be equal: $C_b=C_f= C$.  Then
\eq
\label{e:NbNf1}
    N_b'N_f
    =
    \frac{\det C_f}{\det  C_b}
    =1,
\en
and, with $A=C^{-1}$,  \refeq{SA-bis} becomes
\begin{equation}
\label{e:SA}
    S_{A}
=
    \sum_{u,v\in \Lambda }
    A_{u,v}
    \left(
    \phi_{u}\bar\phi_{v} +
    \psi_{u} \bar\psi_{v}
    \right)
.
\end{equation}
Also, in view of \refeq{NbNf1}, the normalisation constants
cancel in \refeq{ECbf}, which becomes
\begin{align}
    \label{e:ECbfssz}
    \Ex_{\Cbf} K F
    &
    =
    \Big(\int_{\C^{M}}
    K \,e^{-S_b}
    \prod_{k=1}^{M_b}d\bar\phi_k d{\phi}_k \Big)
    \Big(\int_{\Lambdabold'_{f}}
    e^{-S_f}F\Big)
    =I_KJ_F.
\end{align}
The Laplacian \eqref{e:exp-complex-Laplacian} now simplifies to
\begin{equation}
\label{e:Lapss}
    \frac 12
    \Delta_{\Cbf}
=
    \sum_{k,l=1}^{M_b}
    C_{k,l}
    \left(
    \frac{\partial}{\partial \phi_{k}}
    \frac{\partial}{\partial \bar\phi_{l}}
    +
    \frac{\partial}{\partial \psi_{k}}
    \frac{\partial}{\partial \bar\psi_{l}}
    \right)
,
\end{equation}
and from \eqref{e:ECbfssz} we obtain
\begin{equation}
    \Ex_{\Cbf} \bar\phi_k \phi_l = \Ex_{\Cbf} \bar\psi_k \psi_l = C_{kl}.
\end{equation}

\subsection{Factorisation property of the expectation}
\label{sec:facexp}

We now present a factorisation property of the expectation that
is needed in \cite{BS-rg-step}.
We formulate the factorisation property in the supersymmetric
setting of Section~\ref{sec:supersymmetry} for simplicity, although
it does hold more generally.

Let $\Lambda = \Lambdabold_b = \Lambdabold_f$, and let
$X \subset \Lambda$.  We define
$\Ncal (X)$ to be the set of all
$F = \sum_{y \in \Lambda^*} F_y\psi^y\in \Ncal$ such that
$F_{y}=0$ if any component of $y$ is not in $X$,
and such that, for all $y$, $F_y$ does not depend on $\phi_x$ for any
$x \not\in X$.
Similarly, given $X' \subset \Lambda'$,
we define $\Ncal(\Lambdabold \sqcup X')$ as those $F$ that only depend on
the fermion and boson fields indexed by $\Lambda \sqcup X'$.

\begin{prop}\label{prop:factorisationE}
Let $X,Y \subset \Lambda$, let
$F_1(X) \in \Ncal(\Lambdabold \sqcup X')$,
$F_2(Y) \in \Ncal(\Lambdabold \sqcup Y')$, and
suppose that $C_{x',y'}= 0$ whenever
$x' \in X'$,
$y' \in Y'$.
Then the expectation $\Ebold_{\Cbf}$
has the \emph{factorisation property}:
\begin{equation}
\label{e:Efaczz}
    \Ebold_{\Cbf} \big( F_1(X)F_2(Y) \big)
    =
    \big(\Ebold_{\Cbf} F_2(X)\big) \big(  \Ebold_{\Cbf}F_2(Y) \big).
\end{equation}
\end{prop}

\begin{proof}
By linearity of
the expectation, it suffices to consider the case where
$F_1(X)$ is of the form $f_1\psi^{x}$ where $f_1$ depends only on the
boson field in $\Lambdabold \sqcup X'$ and
$x \in  (X')^*$, and where
$F_2(Y)$ is of the form $f_2\psi^{y}$ where $f_1$ depends only on the
boson field in $\Lambdabold \sqcup Y_b'$ and
$y \in  (Y')^*$.
According to \eqref{e:ECbfssz}, the
expectation factors as
\begin{equation}
    \Ebold_{\Cbf} f_1\psi^{x} f_2 \psi^{y}
    =
    (\Ebold_{\Cbf} f_1 f_2 )(\Ebold_{\Cbf} \psi^{x} \psi^{y}),
\end{equation}
where the first expectation on the right-hand side is a
bosonic expectation with covariance matrix $C$, while the
second is a fermionic expectation which is equal to a determinant of a
submatrix of $C$ taken from rows and columns labelled by the
points in $x$ and $y$.

By assumption, the covariance matrix elements vanish for rows
and columns labelled by points in $X$ and $Y$, respectively.  It
is a standard fact that uncorrelated Gaussian random vectors are
independent \cite{Eato07}, and hence $\Ebold_{\Cbf} f_1 f_2
=(\Ebold_{\Cbf} f_1 ) (\Ebold_{\Cbf} f_2 )$.  Also
by assumption, the
determinant yielding the fermion expectation is the determinant of a
block diagonal matrix, so also factors to give $(\Ebold_{\Cbf}
\psi^{x} \psi^{y}) =(\Ebold_{\Cbf} \psi^{x})(\Ebold_{\Cbf}
\psi^{y})$.  This completes the proof.
\end{proof}

\section{The \texorpdfstring{$T_\phi$}{Tphi} semi-norm}
\label{sec:Tphi-props}

\subsection{Motivation}

In the progressive integrations carried out in the renormalisation
group approach, it is necessary to estimate how the size of the result
of an integration compares with the size of the integrand.  When
integrating real-valued functions of real variables, the inequality
\begin{equation}
\label{e:intav}
    \left| \int f(x) dx \right| \le \int |f(x)| dx
\end{equation}
is fundamental.  We need an analogue of \refeq{intav} for the
Gaussian integral $\Ex_{\Cbf} : \Ncal(\Lambdabold \sqcup \Lambdabold')
\to \Ncal(\Lambdabold)$.  In particular, we need to define norms
(or semi-norms) so
that $\Ncal(\Lambdabold \sqcup \Lambdabold')$ and $\Ncal(\Lambdabold)$
become normed algebras.  The norms we define here emerge from a long
history going back to \cite{BY90}; other norms in the purely fermionic
context are developed in \cite{FKT04}.

We choose $\Lambdabold_b$ and $\Lambdabold_f$ each to consist of
disjoint unions of copies of the discrete $d$-dimensional torus of
side length $mR$,
namely
\begin{equation}
\label{e:RLambda}
    \Lambda = \Zd / (mR\Zd),
\end{equation}
where $R\ge 2$ and $m \ge 1$ are integers.  As a basic example,
suppose there are two species of field: the first species is a complex
boson field as in Section~\ref{sec:cbf}, and the second species is a
conjugate fermion field as in Section~\ref{sec:cff}. We choose
$\Lambdabold_b= \Lambda_{1} \sqcup \bar\Lambda_{1}$ and
$\Lambdabold_f=\Lambda_{2} \sqcup \bar\Lambda_{2}$, where each
$\Lambda_{i}$ and $\bar{\Lambda}_{i}$ is a copy of $\Lambda$. The
fermion field $(\psi_v)_{v \in \Lambda}$ is $\psi_{y \in
\Lambdabold_{f}}$ restricted to $y \in \Lambda_{2}$, the fermion field
$(\bar\psi_v)_{v \in \Lambda}$ is $\psi_{y \in \Lambdabold_{f}}$
restricted to $y \in \bar\Lambda_{2}$, and the complex boson field
$(\bar\phi_v,\phi_v)_{v \in \Lambda}$ is the restriction of $\phi$ to
$\Lambda_{1} \sqcup \bar\Lambda_{1}$. For $x \in \Lambda_{1}$ let
$\bar{x}$ be the corresponding point in the copy $\bar\Lambda_{1}$.
The restriction of $\phi$ to $\Lambda_{1} \sqcup \bar\Lambda_{1}$ is a
complex field $\phi=u+iv$ as defined in Section~\ref{sec:cbf} if and
only if
\begin{equation}
    \label{e:reality-condition}
    \phi_{\bar{x}}=\bar\phi_{x},
\quad\quad
    x \in \Lambda_{2}
.
\end{equation}

Given $a \in \R$ and $u \in \Lambda$, an example of an element of
$K\in \Ncal(\Lambdabold)$ is given by
\begin{equation}
\label{e:eatau}
    K(\phi,\bar\phi) =e^{-a (\phi_u\bar\phi_u + \psi_u\bar\psi_u)}.
\end{equation}
Functions of the fermion field are defined as elements of $\Ncal$ via
Taylor expansion in powers of the fermion field.  Due to
anti-commutativity and the finite index set for the fermion field,
such Taylor series always truncate to polynomials in the fermion
field.  For \refeq{eatau}, the Taylor polynomial is
\begin{equation}
    K(\phi)=e^{-a (\phi_u\bar\phi_u + \psi_u\bar\psi_u)}
    =
    e^{-a \phi_u\bar\phi_u}\left( 1 - a\psi_u\bar\psi_u\right).
\end{equation}
For functions of products of even numbers of $\psi$ factors, which are
the only kind we will encounter, there is no sign ambiguity in the
Taylor expansion.

We also consider Taylor expansion in the boson field.  For this, we
replace $\phi$ by $\phi+\xi$ and expand in powers of $\xi$.  We use
the set $\Lambda \sqcup \bar\Lambda$ to keep track of factors $\xi$
versus $\bar\xi$, by writing, e.g., $\xi^x = \xi_{x_1} \xi_{x_2}
\bar\xi_{\bar x_3} \bar\xi_{\bar x_4}$ for $x =
(x_1,x_2,\bar{x}_3,\bar{x}_4)$, and similarly for the fermion field.
A general $K \in \Ncal(\Lambdabold)$ then has (formal) Taylor
expansion
\begin{equation}
\label{e:norm-mot1}
    K(\phi+\xi)
    \sim
    \sum_{ x,y}
    \frac{1}{x!y!}
    K_{x,y}(\phi)
    \xi^x \psi^y,
\end{equation}
where the sum is over sequences $x\in(\Lambda \sqcup \bar\Lambda)^*$
and $y\in(\Lambda \sqcup \bar\Lambda)^*$, and where the coefficients
$K_{x,y}$ are symmetric in the elements of $x$ and anti-symmetric in
the elements of $y$.

Given $\phi$, we define the semi-norm of $K$, in terms of the
coefficients $K_{x,y}(\phi)$.  These coefficients eventually vanish
once the sequence $y$ has length exceeding twice the cardinality of
$\Lambda$.  In general, the coefficients will be non-zero for
infinitely many values of $x$, but the semi-norm will examine only
those with length of $x$ at most $p_\Ncal$ for a fixed choice of the
parameter $p_\Ncal$ (this replaces the ``formal'' Taylor expansion
above by a Taylor polynomial).  The semi-norm is designed to be used
in conjunction with integration, where fields have a typical size.
This motivates us to define the semi-norm of $K$ to be the result of
replacing $\xi^x \psi^y$, in each term in the truncation of the sum
over $x$ at length $p_\Ncal$ in \refeq{norm-mot1}, by a test function
$g_{x,y}$ whose size and smoothness mimic the behaviour expected for
products of typical fields.

The precise definition of the semi-norm, given below, is more general
than the above in several respects.  It allows the possibility of more
``species'' of field than the boson and fermion fields above, and
allows scalar, complex, and multi-component fields.  It allows
distinction between the size of the test functions in its components
corresponding to different field species, and leaves flexible the
choice of weights governing the test functions.

In the remainder of Section~\ref{sec:Tphi-props}, we define the
$T_\phi$ semi-norm on $\Ncal$ and state and develop
its properties.  Most proofs are deferred to
Sections~\ref{sec:tphisemi}--\ref{sec:integration}.

\subsection{Sequence spaces}
\label{sec:seq}

The sets $\Lambdabold_{b}$ and $\Lambdabold_{f}$ are required to have
the following particular structure.  First, $\Lambdabold_{b}$
decomposes into a disjoint union of sets $\Lambdabold_{b}^{(i)}$, for
$i=1,\ldots,s_b$, corresponding to $s_b$ distinct boson field
\emph{species}.  Each set $\Lambdabold_{b}^{(i)}$ is either $\Lambda
\sqcup \bar\Lambda$ (for a species of complex field) or is the
disjoint union of $c_b^{i}$ copies of $\Lambda$ (for a field species
with $c_b^{(i)}$ real components).  The set $\Lambdabold_f$ has the
same structure, but with a possibly different number $s_f$ of species
which can also have components.  Then, as before, we set $\Lambdabold
= \Lambdabold_b \sqcup \Lambdabold_f$, and $\Lambdabold^*$ is the
corresponding set of sequences.  Each $u \in \Lambdabold$ thus carries
a species label $i =i(u) \in \species = \{1,\ldots, s\}$, where
$s=s_b+s_f$.

Of specific interest is the subset $\vec\Lambdabold^*$ of
$\Lambdabold^*$, which consists of sequences whose species labels are
ordered in such a way that the first elements of $z\in
\vec\Lambdabold^*$ are of species $\Lambdabold_{b}^{(1)}$, the next
are of species $\Lambdabold_{b}^{(2)}$, and so on until the boson
species have been exhausted, and then subsequent elements are first of
species $\Lambdabold_f^{(1)}$, then $\Lambdabold_f^{(2)}$, and so on.
For example, a complex species of boson field has components
$\phi$ and $\bar{\phi}$, which belong to the same species, so entries
$z_{j}$ of $(z_{1},\dots ,z_{p})\in \vec\Lambdabold^*$ are not ordered
according to where they are in $\Lambda \sqcup \bar\Lambda$, likewise
for a fermion species $\psi,\bar{\psi}$. We also define
$\vec\Lambdabold_b^*$ and $\vec\Lambdabold_f^*$ to be the subsets of
$\vec\Lambdabold^*$ consisting of only boson or only fermion species.
There is a canonical bijection between $\vec\Lambdabold^*$ and the
Cartesian product $\Lambdabold_{b}^{(1)*}\times \cdots\times
\Lambdabold_{b}^{(s_b)*} \times \Lambdabold_{f}^{(1)*}\times
\cdots\times \Lambdabold_{f}^{(s_f)*}$, given by the correspondence in
which a single sequence in $\vec\Lambdabold^*$ is regarded as a
collection of subsequences of each species.  We will sometimes blur
the distinction between $\vec\Lambdabold^*$ and the Cartesian product
in what follows.  In $\vec\Lambdabold^*$, concatenation $z'\concat z''$
of two sequences is defined by concatenation of each of the individual
species subsequences.  Then $\vec\Lambdabold^*$ is closed under
concatenation.

For $r \ge 0$, we write $\vec\Lambdabold^{(r)}$ for the
subset of $\vec\Lambdabold^*$ consisting of sequences of length $r$,
with the degenerate case $\vec\Lambdabold^{(0)}=\{\varnothing\}$.

\subsection{Test functions}
\label{sec:tf}

Recall from \eqref{e:RLambda} that $\Lambdabold$ is a disjoint union
of copies of a lattice torus.  A \emph{test function} is a function $g
: \vec\Lambdabold^* \to \C$.  In particular, even when there are
complex fields, no relation such as \eqref{e:reality-condition} is
imposed on test functions.  We will define a norm on the set of test
functions as a weighted finite-difference version of a $\Ccal^{k}$ norm,
where $k$ is however proportional to the number of arguments of
$g$, i.e., the length of the sequence in $\vec\Lambdabold^*$.

First we need notation for multiple finite-difference
derivatives.  We write $\units$ for the set $\{\pm e_1, \ldots, \pm e_d\}$
of $2d$ positive and negative unit vectors in $\Zd$.  For a unit
lattice vector $e \in \units$ and a function $f$ on $\Lambda$ the
difference operator is given by $\nabla^{e} f_{x}=f_{x+e} -
f_{x}$. When $e$ is the negative of a standard unit vector
$\nabla^{e}$ is the negative of a conventional backward derivative.
Derivatives of test functions are defined
as follows.  Let $A = \Nbold_{0}^{\units}$,
and for an integer $r>0$, let $\Acal^{(r)}= A^{r} \times
\vec\Lambdabold^{(r)}$.  In the degenerate case, we set $\Acal^{(0)} =
\{\varnothing\}$.  The operator $\nabla^\varnothing$ is the identity operator,
and for $r>0$, $\alpha=(\alpha_1,\ldots,\alpha_r)\in \Acal^{(r)}$ and $z=(z_1,\ldots,z_r)\in
\vec\Lambdabold^{(r)}$, we define
\begin{equation}
    (\nabla^\alpha g)_z
    =
    \nabla_{z_1}^{\alpha_1} \cdots \nabla_{z_r}^{\alpha_r}g_{z_1,\ldots,z_r}.
\end{equation}
Thus, each $\alpha_k$ is a multi-index which specifies
finite-difference derivatives with respect to the variable $z_k$.

\begin{defn}
\label{def:gnorm-general} Fix $p_\Ncal \in \N_0 \cup \{+\infty\}$, and
consider the set of test functions such that $g_z=0$ whenever $z$ has
more than $p_\Ncal$ boson components.  Let $w:A\times \Lambdabold
\rightarrow [0,\infty]$ be a given function.  For $r>0$ and
$(\alpha,z)\in \Acal^{(r)}$, we write $w_{\alpha,z} = \prod_{k=1}^{r}
w_{\alpha_k,z_k}$, and we set $w_\varnothing =1$ in the degenerate
case $r=0$.  We define the $\Phi$ norm on test functions by
\begin{align}
    \|g\|_{\Phi }
    &=
    \sup_{(\alpha,z) \in \Acal}
    w_{\alpha,z}^{-1} |\nabla^\alpha g_{z}|
.
\end{align}
Let $g^{(r)}:\vec\Lambdabold^{(r)}\to\C$ denote the restriction of $g :
\vec\Lambdabold^* \to \C$ to $\vec\Lambdabold^{(r)}$.  The $\Phi$ norm
induces the $\Phi^{(r)}$ norm on these restricted test functions by
\begin{align}
    \|g^{(r)}\|_{\Phi^{(r)}  }  & =
    \sup_{(\alpha,z)\in \Acal^{(r)}}
    w_{\alpha,z}^{-1} |\nabla^\alpha g^{(r)}_{z}| ,
\end{align}
with $0^{-1}=\infty$ and $\infty^{-1}=0$, and with
\begin{align}
    \|g\|_{\Phi }
    &=
    \sup_{r \ge 0}\|g^{(r)}\|_{\Phi^{(r)}  }.
\end{align}
When it is important to make the dependence on $w$ explicit we write
$\Phi(w)$ and $\Phi^{(r)}(w)$.
\end{defn}

As an instance of restriction, suppose that that there is just one
species of field, namely a single
complex boson field $(\bar\phi_x,\phi_x)_{x\in\Lambda}$.
We may regard this field as a test function by extending it to be the zero
function on sequences in $\vec\Lambdabold^*$ of length different from
$1$.  This special case will frequently be relevant for us.

\begin{example}
\label{ex:h} Fix any integer $p_\Phi \ge 0$ and for each species $i$
fix $\h_i >0$.  Let $R$ be the constant of \refeq{RLambda}.  In
applications, the period of the torus is $L^N$ for integers $L,N>1$,
the torus can thus be paved by disjoint blocks of side length $L^j$
for $j=1,\ldots, N$, and we take $R=L^j$ (so $m=L^{N-j}$ to give
$mR=L^N$).  The choice of weight $w:A\times \Lambdabold \rightarrow
[0,\infty]$ given by
\begin{equation}
    w_{\alpha_k,z_k}^{-1} =
    \begin{cases}
    \h_{i(z_k)}^{-1} R^{\alpha_k} & \text{if $|\alpha|_1 \le p_\Phi$}
    \\
    0 & \text{if $|\alpha|_1 > p_\Phi$},
    \end{cases}
\end{equation}
defines the normed space $\Phi(\h)$.  We have written
$R^{\alpha_{k}}=R^{|\alpha_k|_{1}}$ where $|\alpha_k|_1$ is the order
of the derivative $\nabla^{\alpha_k}$.  Then test functions in the
unit ball $B(\Phi)$ of $\Phi$ are those which obey the estimate
\begin{equation}
\lbeq{tfe}
    |\nabla^\alpha g_z| \le \h^{z}R^{-\alpha},
\end{equation}
for all $z$ with at most $p_\Ncal$ boson components, and for all
$\alpha$ with $|\alpha_{k}|_{1}\le p_{\Phi}$ for each component
$\alpha_{k}$ of $\alpha \in \Acal $.  Here $\h^{z}$ is an abbreviation
for $\prod_k \h_{i(z_k)}^{z_j}$.  The estimate \refeq{tfe} means that
$g$ is approximately constant on regions whose diameter is small
compared to $R$. Note that the parameter $p_{\Phi}$ specifies that
$p_{\Phi}$ derivatives per argument of $g$ are bounded by the norm,
whereas the parameter $p_{\Ncal}$ is an upper bound on how many
bosonic spatial variables a test function can depend on.
\end{example}

Let $\vec r \in \N_0^s$ and let $\Phi^{(\vec r)}$ denote the restriction
of $\Phi$ to test functions defined on the subset of $\vec\Lambda^*$
consisting of sequences with exactly $r_i$ components of species $i$
for each $i=1,\ldots,s$.
Given $\vec r$, $g' \in \Phi^{(\vec r)}$, and $g'' \in
\Phi$, we define $g \in \Phi$
by setting $g_{z} = g'_{z'}g''_{z''}$ for $z=z'\concat
z''$, with $g_z=0$ whenever $z$ has fewer than $r_i$ elements
of species $i$ for any $i$.
It follows from the definition of the norm that
\begin{equation}\label{e:plusnormass}
    \|g\|_{\Phi}
    \le
    \|g'\|_{\Phi^{(\vec r)}}
    \|g''\|_{\Phi},
\end{equation}
and we will use this fact later.
Here it is
the fact that $g' \in \Phi^{(\vec r)}$ which provides a unique decomposition
$z=z'\concat z''$ to make $g$ well defined.
A similar inequality is obtained whenever a unique decomposition is
specified.  For example, suppose that we designate some field species
as prime species and some as double prime.  Then $z$ can be decomposed
in a unique way as $(z',z'')$ and if we define a test function $g$
by $g_z = g'_{z'}g''_{z''}$, then it follows from the definition of
the norm that
\begin{equation}\label{e:plusnormass-2}
    \|g\|_{\Phi}
    \le
    \|g'\|_{\Phi}
    \|g''\|_{\Phi}.
\end{equation}

\subsection{Definition of the \texorpdfstring{$T_\phi$}{Tphi} semi-norm}
\label{sec:Tphidef}

Given $F\in \Ncal(\Lambdabold)$,
$x=(x_1,\ldots, x_p)\in \Lambdabold_{b}$, $y \in \Lambdabold_f$,
and a boson field $\phi \in \R^{\Lambdabold_b}$, we
write
\begin{equation}
\lbeq{Fxyphi}
    F_{x,y}(\phi) = \frac{\partial^p F_y(\phi)}{\partial \phi_{x_p} \cdots
    \partial \phi_{x_1}}
.
\end{equation}
(This notation is consistent with Definition~\ref{def:Npoly}.)  We are
writing the boson field as an element of $\R^{\Lambdabold_b}$ for
simplicity, but our intention is to include the possibility of complex
species and for such species derivatives are with respect to
$\phi_{x_i}$ or $\bar\phi_{x_i}$ depending on whether $x_i$ is an
element of $\Lambda$ or $\bar\Lambda$.  This point will be made more
explicit in Section~\ref{sec:Tphi-ex} below.  Also, for $x \in
\vec\Lambdabold_{b}^*$, we write $x!=\prod_{i=1}^{s_{b}} x_i!$ where
the product is over species and $x_i!$ denotes the factorial of the
length of the species-$i$ subsequence of $x$. Similarly $y!$ is
defined for $y \in \vec\Lambdabold_{f}^*$.  For $z = (x,y)$ in
$\vec\Lambdabold_{b}^*$ we write $z!= x!y!$.

\begin{defn}
\label{def:Tphi-norm}
For a test function $g:\vec\Lambdabold^{*}
\rightarrow \Cbold$,
for $F \in \Ncal (\Lambdabold)$, and for $\phi \in \R^{\Lambdabold_{b}}$,
we
define the \emph{pairing}
\begin{equation}
\label{e:Kgpairdef}
    \langle F , g \rangle_\phi
    =
    \sum_{z\in \vec\Lambdabold^{*}}
    \frac{1}{z!}  F_{z}(\phi) g_{z}
    =
    \sum_{x\in \vec\Lambdabold_b^{*}} \sum_{y\in \vec\Lambdabold_f^{*}}
    \frac{1}{x!y!}  F_{x,y}(\phi) g_{x,y}
,
\end{equation}
and the $T_\phi$ \emph{semi-norm}
\begin{equation}
\label{e:Tdef}
    \|F\|_{T_\phi}
    = \sup_{g \in B (\Phi)}
    | \langle F , g \rangle_\phi |
,
\end{equation}
where $B(\Phi)$ denotes the unit ball in the space $\Phi$ of test
functions.
\end{defn}

By definition, $F_{x,y}$ is symmetric under permutations
within each subsequence of $x$ having the same species, and is
similarly antisymmetric in $y$.  This symmetry is reflected by
a corresponding property of the pairing.  To develop this idea,
we begin with the following definition.

\begin{defn}
\label{def:S}
For $z \in \vec\Lambdabold^{(r)}$, let $\vec\Sigma_z$ denote the
set of permutations of $1,\ldots, r$ that preserve the order of
the species of $z$.  For $\sigma \in \vec\Sigma_z$ we define
$\sigma z \in \vec\Lambdabold^{(r)}$ by $(\sigma z)_i=z_{\sigma(i)}$, and
we use this to define a map $S : \Phi \to \Phi$ by
\eq
    (Sg)_z = \frac{1}{z!} \sum_{\sigma \in \vec\Sigma_z}
    \sgn(\sigma_f)  g_{\sigma z},
\en
where $\sigma_f$ denotes the restriction of $\sigma$ to the fermion
components of $z$ and $\sgn(\sigma_f)$ denotes the sign of this permutation.
\end{defn}

\begin{prop}
\label{prop:pairingS}
For $F \in \Ncal(\Lambdabold)$,  $g \in \Phi$,
and $\phi \in \R^{\Lambdabold_{b}}$,
\eq
    \pair{F,g}_\phi = \pair{F,Sg}_\phi.
\en
\end{prop}

\begin{proof}
By the above-mentioned symmetry,
$F_{z}(\phi) = \sgn(\sigma_f)F_{\sigma(z)}(\phi)$
for all $\sigma \in \vec\Sigma_z$.  This implies that
$F_z(\phi) = \frac{1}{z!} \sum_{\sigma \in \vec\Sigma_z}
\sgn(\sigma_f)F_{\sigma(z)}(\phi)$, and hence
\begin{align}
    \pair{F,g}_\phi
    & =
    \sum_{z\in \vec\Lambdabold^{*}}
    \frac{1}{z!}  F_{z}(\phi) g_{z}
     =
    \sum_{z\in \vec\Lambdabold^{*}}
    \frac{1}{z!}  \frac{1}{z!} \sum_{\sigma \in \vec\Sigma_z}
    \sgn(\sigma_f)F_{\sigma(z)}(\phi) g_{z}.
\end{align}
The sum over $z$ is graded by sums over sequences of fixed length
and species choices, and for $z$ fixed within this gradation the set
$\vec\Sigma_z$ is independent of $z$.
It therefore makes sense to replace the summand
within the sum over $\sigma$ by an equivalent expression with
$z$ replaced by $\sigma^{-1}z$, and this does not change the sum.
This gives
\begin{align}
    \pair{F,g}_\phi
    & =
    \sum_{z\in \vec\Lambdabold^{*}}
    \frac{1}{z!}  \frac{1}{z!} \sum_{\sigma \in \vec\Sigma_z}
    \sgn(\sigma_f)F_{z}(\phi) g_{\sigma^{-1}z}
    =
    \sum_{z\in \vec\Lambdabold^{*}}
    \frac{1}{z!}  F_{z}(\phi)
    \frac{1}{z!} \sum_{\sigma \in \vec\Sigma_z}
    \sgn(\sigma_f) g_{\sigma^{-1}z}.
\end{align}
Since $\sgn(\sigma_f)=\sgn(\sigma^{-1}_f)$, and since summing
over $\sigma$ is the same as summing over $\sigma^{-1}$, this gives
the desired result.
\end{proof}

\begin{example}
\label{ex:pairing}
As a simple example of the zero-field pairing,  for fixed points
$x_i \in\Lambda$ and for $p \le p_\Ncal$, let
$F(\phi) = \prod_{i=1}^p \nabla^{\alpha_i}\phi_{x_i}$.
Direct computation
shows that
\refeq{Kgpairdef} leads to
\eq
    \langle F, g \rangle_0
    =
    \nabla^{\alpha_1}_{x_1} \cdots
    \nabla^{\alpha_p}_{x_p} \, (Sg)_{x_1,\ldots,x_p}.
\en
The right-hand side is in general
not the same as the corresponding expression
with $S$ omitted.  This shows that the pairing has a symmetrising
effect.
\end{example}

By definition,
\begin{equation}
\label{e:Tphi-r}
    \|F\|_{T_\phi}
    = \sup_{r \ge 0} \sup_{g^{(r)} \in B (\Phi^{(r)})}
    | \langle F , g^{(r)} \rangle_\phi |
    .
\end{equation}
Note that $\|F\|_{T_\phi}$ is always at least as large as
$|F_{\varnothing}|$ because this is the contribution from
the empty sequence part $g_{\varnothing}$ of the test function,
corresponding to $r=0$. The
$T_\phi$ semi-norm has several attractive and useful properties.  The most
fundamental of these is the product property stated in the following
proposition.  Its proof is given in Section~\ref{sec:pfprod} below.

\begin{prop}
\label{prop:prod}
For  $F,G \in \Ncal$,
$\|FG\|_{T_\phi} \leq \|F\|_{T_\phi}\|G\|_{T_\phi}$.
\end{prop}

Another property is the following
proposition, which is proved in Section~\ref{sec:ebdTphi} below.
In its statement, $e^{-F}$ is defined by Taylor expansion in the
fermion field.  In general, this can introduce sign ambiguities,
but the semi-norm is insensitive to these by \refeq{Tphi-r}.  However,
in our application in \refeq{bdFfoptbis} below, no sign ambiguity
arises.

\begin{prop}
\label{prop:eK}
Let $F\in \Ncal$ and let $F_\varnothing$ be the
purely bosonic part of $F$. Then
\eq
    \|e^{-F  }\|_{T_\phi}
    \le
    e^{-2{\rm Re} F_\varnothing(\phi)+\|F  \|_{T_\phi}}
    .
\en
\end{prop}

\subsection{Example for the \texorpdfstring{$T_\phi$}{Tphi} semi-norm}
\label{sec:Tphi-ex}

For the next proposition, we consider the
case $\Lambdabold_b = (\Lambda \sqcup \bar\Lambda)$
and $\Lambdabold_f = (\Lambda \sqcup \bar\Lambda)$,
corresponding to a complex boson field
$(\bar\phi,\phi)$
and a conjugate fermion field $(\bar\psi,\psi)$.
We use the test function
space $\Phi(\h)$ of Example~\ref{ex:h}, with its associated space
$T_\phi(\h)$, where $\h$ takes the same value for all fields.
For a complex boson field $\phi$ and $x \in \Lambda$,
we define $\tau_x \in \Ncal$ by
\eq
    \tau_x = \phi_x\bar\phi_x + \psi_x\bar\psi_x.
\en
We may regard $\phi_x$ as an element of $\Ncal$.  By definition
its $T_\phi$ semi-norm is $\|\phi_x\|_{T_\phi(\h)} = |\phi_x|+\h$.
We may also regard the boson field $\phi$ as the test function obtained
by extending to the zero function on sequences in $\vec\Lambdabold^*$
which do not consist of a single component in $\Lambdabold_b$;
then its norm is $\|\phi\|_\Phi$.
Since $\h^{-1}|\phi_x| \le  \|\phi\|_{\Phi(\h)}$
by definition, we have
\eq
    \|\phi_x\|_{T_\phi(\h)} \le \h \left(1+ \|\phi\|_{\Phi(\h)} \right).
\en

\begin{prop}
\label{prop:taunorm}
The $T_\phi(\h)$ semi-norm of $\tau_x$ obeys the identity
\eq
\label{e:taunorm}
    \|\tau_x\|_{T_\phi(\h)} = (|\phi_x|+\h )^2+\h^2
\en
and the inequality
\eqalign
\label{e:VxT0}
    \|\tau_x\|_{T_\phi(\h)}
    &
    \le 3\h^2 (1+\|\phi\|_{\Phi(\h)}^2)  .
\enalign
Suppose that $a \in \C$ obeys
$|{\rm Im}\, a| \le \frac 12 {\rm Re}\, a$.
Given any real number $q_2$, there is a
constant $q_{1}$ (with $q_{1}=O(q_2^2)$ as $q_2\to\infty$) such that
\eqalign
\label{e:bdFfoptbis}
    \|e^{-a\tau_x^2}\|_{T_{\phi(\h)}}
    & \le
    e^{({\rm Re}\,a)\h^4 (q_1-q_2 |\phi_x/h|^2 )}.
\enalign
\end{prop}

\begin{proof}
By definition,  $\tau_x = \phi_x\bar\phi_x + \psi_x\bar\psi_x$.
Also by definition, the semi-norm of a sum of terms of
different fermionic degree is the sum of the semi-norms, and hence
\eq
    \|\tau_{x}\|_{T_\phi(\h)} = \|\phi_x\bar\phi_x\|_{T_\phi(\h)}
    + \|\psi_x\bar\psi_x\|_{T_\phi(\h)}.
\en
By definition of the semi-norm,
\eq
    \|\psi_x\bar\psi_x\|_{T_\phi(\h)} = \h^2
\en
and
\eq
    \|\phi_x\bar\phi_x\|_{T_\phi(\h)}
    =
    |\phi_x|^2 + |\phi_x|\, \h + \h|\bar\phi_x| + \h^2
    =
    (|\phi_x| + \h)^2.
\en
This proves \refeq{taunorm}.
We write $t= |\phi_x|/\h$ and $P(t)=(t+1)^2+1$.  Then
\eq
    \|\tau_x\|_{T_\phi(\h)} = \h^2 P(t).
\en
Since $t \le \|\phi\|_{\Phi(\h)}$
and $P(t) \le 3(1+t^2)$, this gives
\eq
\label{e:tauxbd}
    \|\tau_x\|_{T_\phi(\h)}
    \le
    \h^2 P(t)
    = \h^2 P(\|\phi\|_{\Phi(\h)})
    \le
    3\h^2 (1+\|\phi\|_{\Phi(\h)}^2),
\en
which proves \refeq{VxT0}.

Let $\alpha = {\rm Re}\, a$.
By \refeq{tauxbd}, the product property,
and the fact that $|a| \le \frac 32 \alpha$ by assumption,
\eq
\label{e:tau2bd}
    \|a\tau_x^2\|_{T_{\phi(\h)}}
    \le
    |a|\, \h^4 P(t)^2
    \le
    \frac 32 \alpha \h^4 P(t)^2
    .
\en
By Proposition~\ref{prop:eK} and \refeq{tau2bd},
\begin{equation}
    \|e^{-a \tau_{x}^2}\|_{T_{\phi(\h)}}
    \le
    e^{-2\alpha |\phi_x|^4}
    e^{\frac 32 \alpha \h^4 \poly(t)^2}
    \leq
    e^{\alpha\h^4 [-2 t^4 + \frac 32 \poly(t)^2]}.
\end{equation}
Since $\poly$ has leading term $t^2$,
given any real number $q_2$ there is a
constant $q_{1} = O(q_2^2)$ such that
$-2 t^4 + \frac 32 \poly(t)^2   \leq q_{1} - q_2 t^2$.
(In fact, a quartic bound also holds, but this quadratic bound
will suffice for our needs.)
This  gives
\refeq{bdFfoptbis}, and completes the proof.
\end{proof}

On the right-hand side of \eqref{e:VxT0}, the appearance of the norm
$\|\phi\|_{\Phi (\h)}$ could be considered alarming, as this involves
a supremum over the entire lattice and typical fields will be
uncontrollably large in some regions of space.  In our applications
this difficulty will be overcome as follows.  First, we need some
definitions.  For $X \subset \Lambda$ and any test function space
$\Phi$, we define a new norm on $\Phi$ by
\begin{align}
\label{e:PhiXdef}
    \|g\|_{\Phi(X)}
    &=
    \inf \{ \|g -f\|_{\Phi} :
    \text{$f_{z} = 0$
    if all components of $z\in\vec\Lambdabold^*$ are in $X$}\}.
\end{align}
As in Section~\ref{sec:facexp}, we define
\begin{align}
\label{e:NXdef}
    \Ncal (X) &= \{ F \in \Ncal :
    F_{z}=0 \; \text{if any component of $z\in\vec\Lambdabold^*$
    is not in $X$}\}
.
\end{align}
Then $\Ncal(X)$ is a subspace of $\Ncal$, and $\Ncal = \Ncal
(\Lambdabold)$.
Suppose now that $F \in \Ncal(X)$.  Changing the
value of $\phi_x$ for $x \not\in X$ has no effect on the pairing of
$F$ with any test function $g$ and hence has no effect on any $T_\phi$
semi-norm of $F$. Thus, returning to \eqref{e:VxT0}, by taking the infimum
over all possible redefinitions of $\phi$ off $X = \{x \}$, we can
replace \eqref{e:VxT0} by
\begin{equation}
    \|\tau_x\|_{T_\phi(\h)}
    \le 3\h^2 (1+\|\phi\|_{\Phi(X,\h)}^2)
.
\end{equation}

\subsection{Further properties of the \texorpdfstring{$T_\phi$}{Tphi}
semi-norm}
\label{sec:Tphifurther}

Recall the definition of polynomial elements of $\Ncal$
in Definition~\ref{def:Npoly}.
The following proposition bounds the
$T_\phi$ semi-norm of a polynomial in terms of the $T_0$ semi-norm.

\begin{prop}\label{prop:T0K}
If $F$ is a polynomial of degree $A \le p_\Ncal$ then
\begin{align}
    \|F\|_{T_{\phi}}
    &\le
    \|F\|_{T_{0}}\big(1+\|\phi\|_{\Phi}\big)^{A}
    .
\end{align}
\end{prop}

It is an immediate consequence of
Proposition~\ref{prop:T0K} that for $A \ge 0$
and any $\kappa \in (0,2^{-1/2}]$,
\begin{align}
\label{e:Fpolyexp}
    \|F\|_{T_{\phi}}
    &\le
    \|F\|_{T_{0}} A^{A/2}\kappa^{-A} e^{\kappa^2 \|\phi\|_{\Phi}^{2}}.
\end{align}
For $A=0$ this is trivial (with $0^{0}=1$),
since then $F$ is simply a complex number $w$ and
$\|F\|_{T_{\phi}} =\|F\|_{T_{0}}=|w|$.
Also, for $A \ge 1$ and $\kappa \in (0,2^{-1/2}]$,
\refeq{Fpolyexp} follows from Proposition~\ref{prop:T0K}
together with the inequality
\begin{align}
    \label{e:expkappa}
    1+x \le \sqrt{2} \big(1+x^{2}\big)^{1/2}
    \le
    A^{1/2}\kappa^{-1} \big(1+2A^{-1}\kappa^{2}x^{2}\big)^{1/2}
    \le
    A^{1/2}\kappa^{-1} e^{A^{-1}\kappa^{2}x^{2}}.
\end{align}

Suppose we have two test function spaces $\Phi$ and $\Phi'$, with
corresponding semi-norms $T_\phi$ and $T_\phi'$.  For $n \ge 0$, let
\eq
\label{e:rhodef1}
    \rho^{(n)} = 2\sup_{r \ge n}
    \sup_{g \in B (\Phi'^{(r)})}
    \|g\|_{\Phi^{(r)}}.
\en
In our applications, $\rho^{(n)}$ will be small for $n \ge 1$.
The following proposition
relates the $T_\phi$ and $T_\phi'$ semi-norms.

\begin{prop}
\label{prop:Tphi-bound}
Let $A< p_\Ncal$ be a non-negative integer
and let $F \in\Ncal$.  Then
\begin{align}
    \|F\|_{T_{\phi}'}
    & \leq
    \left(1 + \|\phi\|_{\Phi'}\right)^{A+1}
    \left( \|F\|_{T_{0}'}
    +
    \rho^{(A+1)} \sup_{0\le t\le 1}
    \|F\|_{T_{t\phi}}  \right).
\label{e:Tphicor1}
\end{align}
\end{prop}

Recalling the discussion around \eqref{e:PhiXdef}, we can improve \refeq{Tphicor1}
by taking the infimum over all possible redefinitions of $\phi$
off $X$, with the result that
\begin{align}
    \|F\|_{T_{\phi}'}
    & \leq
    \left(1 + \|\phi\|_{\Phi'(X)}\right)^{A+1}
    \left( \|F\|_{T_{0}'}
    +
    \rho^{(A+1)} \sup_{0\le t\le 1}
    \|F\|_{T_{t\phi}}  \right)
    \quad
    \text{for $F \in \Ncal(X)$}
.
\label{e:TphicorX}
\end{align}

Finally, the following proposition shows that the map $\theta$ of
Definition~\ref{def:theta-new} has a contractive property.  For its
statement, let $\Lambdabold$, $\Lambdabold'$ and the map $z \mapsto
z'$ be as described above Definition~\ref{def:theta-new} and let
$w:A\times \Lambdabold\rightarrow [0,\infty]$ and $w':A\times
\Lambdabold' \rightarrow [0,\infty]$ be weights as specified in
Definition~\ref{def:gnorm-general}. These weights together define a
new weight $w\sqcup w':A\times
(\Lambdabold\sqcup\Lambdabold')\rightarrow [0,\infty]$. Species in
$\Lambdabold$ and species in $\Lambdabold'$ are distinct, and we order
the species in such a way that a species from $\Lambdabold'$ occurs
immediately following its counterpart in $\Lambdabold$.  We denote the
corresponding norm on test functions
$g:(\overrightarrow{\Lambdabold\sqcup\Lambdabold'})^{*} \to\C$ by
$\Phi (w\sqcup w')$.  Also, we define the function $w+w'$ from $A
\times \Lambdabold$ to $\C$ by
\begin{equation}
    (w + w') (a,z) = w(a,z) + w'(a,z').
\end{equation}

\begin{prop}
\label{prop:derivs-of-tau-bis}
For $F \in \Ncal(\Lambdabold)$,
\begin{align}
    \label{e:theta-bd1}
    \|\theta F\|_{T_{\phi \sqcup \xi} (w\sqcup w')}
&
    \le
    \|F\|_{T_{\phi +\xi} (w+w')}.
\end{align}
\end{prop}

Proofs of Propositions~\ref{prop:T0K}--\ref{prop:derivs-of-tau-bis}
are given in Sections~\ref{sec-Tphiestimates}--\ref{sec:compsp0}  below.

\subsection{Field regulators and associated norms}
\label{sec:fran}

\begin{defn}
\label{def:blocks}
(a)
The set $\Lambda = \Zd/(mR)$ is paved in a natural way by disjoint
cubes of side $R$.  We call these cubes \emph{blocks} and denote the
set of blocks by $\Bcal$.
\\
(b)
A union of blocks is called a \emph{polymer},
and the set of polymers is denoted ${\cal P}$.
The size $|X|_R$ of $X\in {\cal P}$ is the number
of blocks in $X$.
\\
(c)
A polymer $X$ is
\emph{connected} if for any two points $x_{a}, x_{b}\in X$
there exists a path $( x_0,\dotsc ,x_n)$ in $X$
with $\|x_{{i+1}}-x_{i}\|_\infty =1$,
$x_{0} = x_{a}$ and $x_{n}=x_{b}$.
\\
(d)
A polymer
$X\in\Pcal$ is a
\emph{small set} if $X$ is connected and
$|X|_R \le 2^{d}$.
Let $\Scal\subset\Pcal$ be the set of all small sets.
\\
(e)
The \emph{small set neighbourhood} of $X \subset \Lambda $  is
the subset $X^\Box$ of $\Lambda$ given by
\begin{equation}
\label{e:9ssn}
    X^{\Box}
=
    \bigcup_{Y\in \Scal :X\cap Y \not =\varnothing } Y.
\end{equation}
(Other papers have used the notation $X^*$ in place of $X^\Box$,
but we use $X^\Box$ to avoid confusion with our notation
for sequence spaces.)
\end{defn}

Note that, by definition, $X \subset X^\Box$ and $(X\cup
Y)^\Box=X^\Box \cup Y^\Box$.  The following definitions involve a
positive parameter $\ell$ whose value will be chosen to satisfy the
(related) hypotheses of Propositions~\ref{prop:EK}--\ref{prop:EG2}
below.  For concreteness, in these definitions we consider only the
case where $\Lambdabold_b = \Lambda \sqcup \bar\Lambda$ and the boson
field is the complex field of Section~\ref{sec:cbf}.  For the
$n$-component $|\varphi|^4$ model studied in \cite{BBS-phi4-log}, the
same definitions apply with $\phi$ replaced by $\varphi\in
(\R^n)^\Lambda$.

\begin{defn}
\label{def:ffregulator} Given $X  \subset \Lambda$ and $\phi \in
\C^{\Lambda}$, the \emph{fluctuation-field regulator} is given by
\begin{align}
\label{e:GPhidef}
    G (X,\phi)
    =
    \prod_{x \in X}
    \exp  \left(|B_{x}|^{-1}\|\phi\|_{\Phi (B_{x}^\Box,\ell )}^2 \right)
    ,
\end{align}
where $B_{x}$ is the unique block that contains $x$, and where the
norm on the right-hand side is the $\Phi(\h)$ norm
of Example~\ref{ex:h} with $\h=\ell>0$ and
localised to the small set neighbourhood $B^\Box$ as in
\refeq{PhiXdef}.
We define a norm on $\Ncal(X^\Box)$ by
\begin{equation}\label{e:Gnormdef}
   \| F(X)\|_{G ,\ell }
    =
    \sup_{\phi \in \C^\Lambda}
    \frac{\|F(X)\|_{T_{\phi }(\ell )}}{G (X,\phi)}
    \quad \text{for $F(X) \in \Ncal (X^{\Box})$}.
\end{equation}
Although the norm depends on $X$, we choose not to add a subscript $X$
to the norm to make this dependence explicit.
\end{defn}

For $X \in \Pcal$ the formula \eqref{e:GPhidef} simplifies to
\begin{align}
\label{e:GPhidef2}
    G (X,\phi)
    =
    \prod_{B \in \Bcal (X)} \exp  \|\phi\|_{\Phi (B^\Box,\ell )}^2
,
\end{align}
and the more complicated formula in the definition is
a way to extend this simpler formula to all subsets $X \subset \Lambda$.
A similar remark applies to the next definition.

Suppose that $R$ and $m$ are chosen in such a way that the diameter of
$B^\Box$ is less than $mR$ (e.g., if $m$ is sufficiently large).  We
can then identify $B^\Box$ with a subset of $\Zd$ and use this
identification to define polynomial functions from $B^\Box$ to $\C$.
The \emph{dimension} of such a polynomial $f$, of a single variable,
is defined to be $\frac{d-2}{2}$ plus the degree of $f$.  Let
$d_{\Phipoltil}$ be a fixed non-negative integer.  We define
\begin{equation}
    \Phipoltil(B^\Box)
    =
    \left\{ f \in \C^{\Lambda} \mid
    \text{$f$ restricted to $B^\Box$ is a polynomial of dimension
    at most $d_{\Phipoltil}$}\right\}.
\end{equation}
Then, for $\phi \in \C^{\Lambda}$,
we define the semi-norm
\begin{equation}
\label{e:Phitilnorm}
    \| \phi \|_{\tilde{\Phi}  (B^\Box)}
=
    \inf \{ \| \phi -f\|_{\Phi}  : f \in \Phipoltil (B^\Box)\}.
\end{equation}

\begin{defn}
\label{def:regulator}
Given $X  \subset \Lambda$ and $\phi \in \C^{\Lambda}$,
the \emph{large-field regulator} is given  by
\begin{align}
\label{e:9Gdef}
    \tilde G  (X,\phi)
    =
    \prod_{x \in X}
    \exp \left(\frac 12 |B_{x}|^{-1}\|\phi\|_{\tilde\Phi (B_{x}^\Box,\ell)}^2
    \right)
    .
\end{align}
The factor $\frac 12$, which does not occur in \refeq{GPhidef}, has
been inserted in \refeq{9Gdef} for later convenience.  We define a
norm on $\Ncal(X^{\Box})$ by
\begin{equation}
\label{e:9Nnormdef}
    \|F(X)\|_{\tilde{G},\h}
    =
    \sup_{\phi \in \C^\Lambda}
    \frac{\|F(X)\|_{T_{\phi} (\h)}}
    {\tilde{G}(X,\phi)}
    \quad \text{for $F(X) \in \Ncal (X^{\Box})$},
\end{equation}
where we have made explicit in the notation the fact that the norm on
the left-hand side depends on a parameter $\h$ which may be chosen
to be different from the parameter $\ell$ used for the regulators.
The dependence of the norm on $\ell$ is left implicit.
\end{defn}

It is immediate from the definitions that $G (X, \phi)$
and $\tilde{G} (X, \phi)$ are
increasing in $X$, and that
for all disjoint $X,Y$ and for all $\phi \in \C^\Lambda$,
\eqalign
\label{e:GXYfluct}
    G (X \cup Y, \phi)
    &= G (X, \phi ) G  (Y, \phi ),
\\
\label{e:GXY}
    \tilde{G} (X \cup Y, \phi)
    &= \tilde{G} (X, \phi ) \tilde{G} (Y, \phi ).
\enalign
In addition, for $A \ge 0$ there is a $c_A \ge 1$ such that
for all $t \in [0,1]$,
\begin{align}
    &1=G(X,0)\le G(X,\phi), \quad
    \tilde G (X,t\phi)
    \leq   G^{1/2} (X,\phi),
    \nnb &
\label{e:KKK6}
    \left(1+ \|\phi\|_{\Phi(\ell, X^{\Box}) }\right)^{A+1}
\le
    c_A   G ^{1/2}(X,\phi).
\end{align}
The first two inequalities are valid for $X \subset \Lambda$.  The
third holds for $X \in \Pcal$, and  follows from
\eqref{e:GPhidef2}.  The following proposition extends the product
property to the $G$ and $\tilde G$ norms.

\begin{prop}
If
$X,Y$ are \emph{disjoint} and if $F(X) \in
\Ncal (X^{\Box}) ,\ i=1,2$ and $K(Y) \in \Ncal(Y^\Box)$,
then $F(X)K(Y) \in \Ncal ((X\cup Y)^{\Box})$, and for either
of the $G$ or $\tilde G$ norms \refeq{Gnormdef} and \refeq{9Nnormdef},
\begin{equation}
\|F(X)K(Y)\|  \leq
\|F(X)\| \| K(Y)\| .
\end{equation}
\end{prop}

\begin{proof}
This follows immediately from the product property
Proposition~\ref{prop:prod} for the
$T_\phi$ semi-norm, together with
\refeq{GXYfluct}--\refeq{GXY}.
\end{proof}

By definition,
\begin{equation}
\label{e:T0G}
    \|F\|_{T_0(\ell)} \leq \|F\|_{G,\ell}.
\end{equation}
The following proposition
shows that this inequality can be
partially reversed, at the expense of a term involving a
multiple of $\|F\|_{\tilde{G}}$.  In our application, the ratio
$\ell/\h$ appearing in this term will be small.

\begin{prop}
\label{prop:KKK}
Let $X \in \Pcal$ and $F \in \Ncal(X)$.
For any positive integer $A<p_\Ncal$, there is a constant $c_A$ such that
\begin{equation}
\label{e:KKK1}
    \|F \|_{G ,\ell}
\le
    c_A
    \left(
    \|F \|_{T_{0}(\ell)} +
    \left( \frac{\ell}{\h} \right)^{A+1}
    \|F \|_{\tilde{G} ,\h }
    \right).
\end{equation}
\end{prop}

\begin{proof}
We apply Proposition~\ref{prop:Tphi-bound},
with $T_\phi'=T_\phi(\ell)$ and $T_\phi=T_\phi(\h)$.
Then $\rho^{(n)} = (\ell/\h)^n$ by definition.
It follows from \refeq{9Nnormdef} and \refeq{KKK6} that
\begin{equation}
\label{e:KGGG}
    \|F\|_{T_{t\phi}}
    \leq \|F\|_{\tilde{G} ,\h } \, \tilde G (X,t\phi)
    \leq \|F\|_{\tilde{G} ,\h } \, G^{1/2} (X,\phi).
\end{equation}
We use this in the last term on the right-hand side of \refeq{Tphicor1},
to obtain
\eqalign
    \|F \|_{T_{\phi}(\ell)}
    & \leq
    \left(1+ \|\phi\|_{\Phi (\ell)}\right)^{A+1}
    \left(
    \|F \|_{T_{0}(\ell)}
    +     \left(\frac{\ell}{\h} \right)^{A+1}
    \|F\|_{\tilde G,\h}  G^{1/2} (X,\phi)
    \right)
    .
\label{e:KKK4.5}
\enalign
We then apply \refeq{KKK6},
divide by $G(X,\phi)$,
and take the supremum over $\phi$ to obtain \refeq{KKK1}.
\end{proof}

\subsection{Norm estimates for Gaussian integration}

The following proposition shows that the Laplacian, and in view of
\refeq{ELap} also the
Gaussian integral, are bounded operators on a space of polynomials
in $\Ncal$.
In its statement, we regard $\Cbf$ as a test function in $\Phi$,
by extending the definition above \refeq{LapC} to
$\Cbf_z=0$ for $z \in \vec\Lambdabold^*$ unless the length of $z$ is
$2$ and both components are either in $\Lambdabold_b$ or in
$\Lambdabold_f$, in which case it is given respectively by $\Cbf_{b;z}$
or $\Cbf_{f;z}$.  Then it makes sense to take the norm $\|\Cbf\|_\Phi$.

\begin{prop}
\label{prop:Etau-bound}
If $F \in \Ncal$ is a polynomial of degree at most $A$, with $A \le p_\Ncal$,
then
\begin{equation}
\label{e:DCK}
    \|\Delta_{\Cbf} F\|_{T_{\phi}}
    \le
    A^2 \|\Cbf\|_{\Phi}\,\|F\|_{T_{\phi}}
\end{equation}
and
\begin{equation}
\label{e:eDC}
    \|e^{t\Delta_{\Cbf}} F\|
    \le
    e^{|t| A^{2}\|\Cbf\|_{\Phi}}\,\|F\|_{T_{\phi}}.
\end{equation}
\end{prop}

Note that \refeq{eDC} follows from
$\|e^{t\Delta_{\Cbf}}\| \leq \sum_{n=0}^\infty \frac{1}{n!}
\|t\Delta_{\Cbf}\|^n$
together with
\refeq{DCK}, so it suffices to prove \refeq{DCK}.

In the next proposition,
we restrict to the conjugate fermion field setting of
Section~\ref{sec:cff}, with fields $(\bar\psi_x, \psi_x)_{x \in \Lambda}$.
We extend $C_f$ to a test function in $\Phi(\Lambda)$ by
setting it equal to zero when evaluated on any sequence $z$ except
those where $z$ has length $2$ and both components are in
$\Lambda$.
Then the norm $\|C_f\|_{\Phi(w')}$ makes sense.

\begin{prop}
\label{prop:EK}
In the conjugate fermion field setting of
Section~\ref{sec:cff},
suppose that  the covariance $C_f$ obeys
$\|C_f\|_{\Phi(w')}\le 1$.
If $F \in \Ncal (\Lambdabold \sqcup \Lambdabold')$ then
\eq
\label{e:EKz}
    \| \Ex_{\Cbf} F  \|_{T_\phi(w)}  \le
    \Ex_{\Cbf_b}
    \|F  \|_{\Ttimes_{\phi \sqcup \xi}(w\sqcup w')}
    .
\en
Also, if $F \in \Ncal (\Lambdabold )$ then
\eq
\label{e:EK}
    \| \Ex_{\Cbf} \theta F  \|_{T_\phi(w)}  \le
    \Ex_{\Cbf_b}
    \|\theta F  \|_{\Ttimes_{\phi \sqcup \xi}(w\sqcup w')}
     \le
    \Ex_{\Cbf_b}   \|F  \|_{T_{\phi+\xi}(w+w')} .
\en
\end{prop}

The variable $\xi$, which occurs in \refeq{EK} (and also in \refeq{EG2})
is a dummy variable of integration for $\Ex_{\Cbf_b}$.
Note that the first inequality of \refeq{EK} is an immediate
consequence of \refeq{EKz},
and that the second follows from \refeq{theta-bd1},
so it suffices to prove \refeq{EKz}.  In fact, as we show in
Lemma~\ref{lem:EKzz} below, a stronger statement than \refeq{EKz} holds.
Namely, if $h: \R^{\Lambdabold'_{b}} \to \C$ then
\begin{equation}
\lbeq{EKzh}
    \| \Ex_{\Cbf} h F  \|_{T_\phi(w)}  \le
    \Ex_{\Cbf_b} \left[ | h(\xi)|\,
    \|F  \|_{\Ttimes_{\phi \sqcup \xi}(w\sqcup w')} \right]
    .
\end{equation}

Finally, we have  an estimate for the Gaussian expectation
of the fluctuation-field regulator.

\begin{prop}
\label{prop:EG2} Let $t \ge 0$, $\Econstg >1$, and $X  \subset \Lambda$.
There exists a (small) positive constant $c (\Econstg)$ such
that if $ \|\Cbf_b\|_{\Phi^+ (\ell)}\le c (\Econstg) t^{-1}$, where
the $\Phi^+$ norm is the $\Phi$ norm with $p_\Phi$ replaced by
$p_\Phi+d$, then
\eq
\label{e:EG2}
    0 \leq \Ex_{\Cbf_b}  G^t(X,\xi)  \le \Econstg^{R^{-d}|X|}.
\en
\end{prop}

Proofs of Propositions~\ref{prop:Etau-bound}--\ref{prop:EG2}
are given in Sections~\ref{sec:heat}--\ref{sec:ffr}  below.

\section{Gaussian integration and the heat equation}
\label{sec:Gihe}

In this section, we prove Proposition~\ref{prop:conv}.
The proof uses integration by parts.
For the purely bosonic case,
it is straightforward to apply integration by parts to obtain
\begin{equation}
\label{e:Gibp}
    \Ex_{\Cbf_b} \phi_x f
    =
    \sum_{y\in \Lambdabold_b} \Cbf_{b; x,y}\Ex_{\Cbf_b}
    \frac{\partial f}{\partial \phi_y}
    ,
    \quad
    \quad
    x \in \Lambdabold_b
\end{equation}
where $f$ is any smooth function
such that both sides are integrable.
The following lemma is a fermionic version of \refeq{Gibp}.
Although it is standard (see, e.g., \cite[Proposition~1.17]{FKT02}),
we give the simple proof.

\begin{lemma}
\label{lem:fibp}
For $F \in \Ncal (\Lambdabold)$ and $x \in \Lambdabold_f$,
\begin{equation}
    \Ex_{\Cbf_f} \psi_{x} F
=
    \sum_{y\in\Lambdabold_f} \Cbf_{f;x,y}\,
    \Ex_{\Cbf_f}
    i_{y}F
.
\end{equation}
\end{lemma}

\begin{proof}
By definition,
\begin{equation}
    i_{y}
    S
    =
    \frac{1}{2}\sum_{v\in\Lambdabold_f} \Abf_{f;y,v} \psi_{v}
    -
    \frac{1}{2}\sum_{u\in\Lambdabold_f} \Abf_{f;u,y} \psi_{u}
    =
    \sum_{v\in\Lambdabold_f} \Abf_{f;y,v} \psi_{v}
.
\end{equation}
It suffices by linearity to consider $F$ a product of generators,
and since $i_{y}F$ cannot contain
all generators as factors,
\begin{equation}
    \int_{\Lambdabold_f} i_{y}F =0
.
\end{equation}
By replacing $F$ by $e^{-S}F$, we have
\begin{equation}
    \int_{\Lambdabold_f} \left(i_{y}e^{-S} \right) F
    +
    \int_{\Lambdabold_f} e^{-S}
    \left(i_{y}F\right)
=
    0
.
\end{equation}
This is the same as
\begin{equation}
    \int_{\Lambdabold_f} e^{-S} \left(-\sum_{v} \Abf_{f;y,v} \psi_{v}\right) F
    +
    \int_{\Lambdabold_f} e^{-S}
    \left(i_{y}F\right)
=
    0
.
\end{equation}
By applying the inverse of $\Abf_{f}$ to both sides, we obtain the desired
result.
\end{proof}

The following lemma provides the expression in our context of the
intimate link between Gaussian integration and the heat equation.
In the purely bosonic context, this is a standard fact about
Gaussian random variables.

\begin{lemma} \label{lem:*heat-eq}
For $T>0$ and $F\in \Ncal(\Lambdabold)$ such that $F_{t} =
\Ex_{t\Cbf}\theta F$ is defined for $t<T$,
the differential equation
\begin{equation}\label{e:*heat-eq1}
    \frac{d}{dt} F_{t}
    =
    \frac{1}{2} \Delta_{\Cbf} F_{t}
\end{equation}
holds for $t \in (0,T)$.  Moreover, if $P \in
\Ncal(\Lambdabold)$ is a polynomial of finite degree, then
\begin{equation}
\label{e:*EWick}
    \Ex_{\Cbf} \theta P
    =
    e^{\frac{1}{2} \Delta_{\Cbf}} P.
\end{equation}
\end{lemma}

\begin{proof}
Since the Gaussian expectation factors as in \refeq{ECbf}, to prove
\eqref{e:*heat-eq1} it suffices to consider separately the cases where
$F$ is purely bosonic or purely fermionic.

We first prove \eqref{e:*heat-eq1} in the bosonic case, where $F=f$ is
a smooth function of $\phi$.  The expectation is then a standard
Gaussian integral, and by a change of variables we have
\begin{align}\label{e:*heateq2}
    \frac{d}{dt} F_{t} (\phi)
    &=
    \frac{d}{dt}\Ex_{\Cbf_b} F (\phi + \sqrt{t}\xi)
    =
    \Ex_{\Cbf_b} \sum_{x \in \Lambdabold_{b}}
    F_{x}(\phi + \sqrt{t}\xi) \frac{1}{2\sqrt{t}}\xi_{x}.
\end{align}
To differentiate under the expectation we need to know that the
resulting integrand is integrable.  To see this,
we observe that since $t<T$
there exists $\epsilon >0$ such that $F (\phi +
\sqrt{t}\xi)\exp[\epsilon \sum_{x} \xi_{x}^{2}]$ is integrable. Now
we apply the integration by parts identity \refeq{Gibp}, and the
definition \refeq{LapC} of the Laplacian, to conclude that
\begin{align}\label{e:*heateq2a}
    \frac{d}{dt} F_{t} (\phi)
    &=
    \frac{1}{2}
    \Ex_{\Cbf_b} \sum_{x,y \in \Lambdabold_{b}}
    \Cbf_{b;x,y}
    F_{x,y}(\phi + \sqrt{t}\xi)
    \nnb
    &=
    \frac{1}{2}
    \Ex_{\Cbf_b} \Delta_{\Cbf_b}
    F(\phi + \sqrt{t}\xi)
    =
    \frac{1}{2}  \Delta_{\Cbf_b}
    \Ex_{\Cbf_b}
    F (\phi + \sqrt{t}\xi)
    =
    \frac{1}{2}\Delta_{\Cbf_b}
    F_{t} (\phi).
\end{align}
This proves the bosonic case of \refeq{*heat-eq1}.

For the fermionic case,
we can suppose that $F=\psi^y=\psi_{y_1}\cdots \psi_{y_k}$.
We first note that
\begin{align}
    \frac{d}{dt}
    \theta_{t}  \psi^{y}
&=
    \frac{d}{dt}
    \prod_{j=1}^k \left(\psi_{y_{j}} + t \psi_{y'_{j}} \right)
=
    \sum_{i} (-1)^{i-1}\psi_{y'_{i}}
    \prod_{j\not =i} \left(\psi_{y_{j}} + t \psi_{y'_{j}} \right),
\end{align}
with the factors under the product maintaining their original order.
By definition of $i_x$, this gives
\begin{align}
    \frac{d}{dt}
    \theta_{t}  \psi^{y}
&=
    \sum_{x \in \Lambdabold_f} \psi_{x}\frac{1}{t}i_{x}\left( \theta_{t}\psi^{y} \right)
=
    \sum_{x \in \Lambdabold_f} \psi_{x} \theta_{t}\left(i_{x} \psi^{y} \right)
,
\end{align}
where the sum extends to all $x\in\Lambdabold_f$ because terms with
$x\neq y_j'$ for some $j$ vanish.
With Lemma~\ref{lem:fibp}, we then obtain
\begin{align}
    \frac{d}{dt} \Ex_{\Cbf_f} \theta_{t}F
    &=
    \sum_{x \in \Lambdabold_f}
    \Ex_{\Cbf_f} \psi_{x} \theta_{t} \left(i_{x} F \right)
    \nnb
    &=
    \sum_{x,y \in \Lambdabold_f}
    \Cbf_{f;x,y}
    \Ex_{\Cbf_f}
    i_{y} \theta_{t}\left(i_{x} F \right)
    =
    \sum_{x,y \in \Lambdabold_f}
    \Cbf_{f;x,y}
    \Ex_{\Cbf_f}
    \theta_{t}\left(ti_{y}i_{x} F \right)
,
\end{align}
which is the same as
\begin{align}
\label{e:dEC}
    \frac{1}{t} \frac{d}{dt} \Ex_{\Cbf_f} \theta_{t}F
=
    \Delta_{\Cbf_f}\Ex_{\Cbf_f} \theta_{t} F
.
\end{align}
Writing $\frac{1}{t}\frac{d}{dt} =
2\frac{d}{d (t^{2})}$, and then replacing $t^{2}$ by $t$, we
obtain
\begin{equation}
    \frac{d}{dt}\Ex_{\Cbf_f} \theta_{\sqrt{t}}F
=
    \frac 12 \Delta_{\Cbf_f}\Ex_{\Cbf_f} \theta_{\sqrt{t}} F
.
\end{equation}
It can be verified from the definitions that
$\Ex_{\Cbf_f}\theta_{\sqrt{t}} = \Ex_{t\Cbf_f}\theta$,
and the fermionic case of \refeq{*heat-eq1} follows.

Finally, suppose that $F$ is a polynomial $P$
of finite degree.
By \eqref{e:*heat-eq1}, each of
$P_{t}$ and $e^{\frac{t}{2} \Delta_{\Cbf}} P$
solves the heat equation with
the same initial data. The heat equation is a
finite-dimensional linear system of
ordinary differential equations because $\Lambdabold$ is a finite set and
thus $\Delta_{\Cbf}$ is a linear operator acting on the finite-dimensional
vector space of polynomials in $\phi$ and $\psi$. Therefore solutions
for the heat equation are unique by the standard theory of linear systems,
and \eqref{e:*EWick} follows.
\end{proof}

\begin{proof}[Proof of Proposition~\ref{prop:conv}.]
Since \refeq{ELap} has been proven in \refeq{*EWick},
it suffices to prove \refeq{conv}.

By the first equality of \refeq{ECbf},
it suffices to verify \refeq{conv} individually for $F=f$
and $F=\psi^y$.  For $F=\psi^y$, \refeq{conv} is an immediate
consequence of \refeq{*EWick}.
For $F=f$, the expectation is a standard Gaussian expectation.
Since  finite Borel measures are uniquely
characterised by their Fourier transforms, it suffices to consider
the case $f(\phi)=e^{i\phi \cdot \eta}$ for $\eta \in \R^{\Lambdabold_b}$.
The Fourier transform of a Gaussian measure with covariance $\Cbf_b$
is $e^{-(\eta, \Cbf_b \eta)}$.  Thus, setting $\Cbf_b = \Cbf_{b,1}+
\Cbf_{b,2}$, we have
\eq
    \Ex_{\Cbf_b}\theta f = e^{i\phi \cdot \eta}e^{-(\eta, \Cbf_b \eta)},
\en
and also
\eq
    \Ex_{\Cbf_{b,2}}\theta\left( \Ex_{\Cbf_{b,2}}\theta f \right)
    = \Ex_{\Cbf_{b,2}}\theta\left(e^{i\phi \cdot \eta}e^{-(\eta, \Cbf_{b,1} \eta)}
    \right)
    = e^{i\phi \cdot \eta}e^{-(\eta, \Cbf_{b,2} \eta)}e^{-(\eta, \Cbf_{b,1} \eta)}
\en
The above two right-hand sides are equal,
and \refeq{conv} follows in the bosonic case.  This completes the proof.
\end{proof}

\section{The \texorpdfstring{$T_\phi$}{Tphi} semi-norm}
\label{sec:tphisemi}

We now prove the five propositions stated in
Sections~\ref{sec:Tphidef}--\ref{sec:Tphifurther}: the product
property of Proposition~\ref{prop:prod}, the exponential norm estimate
of Proposition~\ref{prop:eK}, the polynomial norm estimate of
Proposition~\ref{prop:T0K}, the change of norm estimate of
Proposition~\ref{prop:Tphi-bound}, and the contractive bound for the
map $\theta$ of Proposition~\ref{prop:derivs-of-tau-bis}.  Many of the
proofs follow the strategy of writing the $T_\phi$ semi-norm in terms of
the pairing \eqref{e:Kgpairdef} that defines it, and then introducing
an adjoint operation that transfers the desired statement into an
estimate on test functions.

\subsection{Proof of the product property}
\label{sec:pfprod}

In this section, we prove the product property
stated in Proposition~\ref{prop:prod}.  The proof proceeds by first
establishing the product property for a more general algebra with semi-norm,
and then noting that the product property of the $T_\phi$ norm follows as
an instance.

Let $\Hcal$ be the algebra, generated by the fermion field, and over
the ring of formal power series in indeterminates
$(\xi_{x})_{x\in \Lambdabold_{b} }$. An element $A\in\Hcal$ has
a unique representation
\begin{equation}
    \label{e:Arep}
    A=\sum_{z \in \vec\Lambdabold^{*}}
    \frac{1}{z!} F_{z} \xi^{z_{b}}\psi^{z_{f}}
,
\end{equation}
where $z=(z_b,z_f)$, the coefficients $F_{z}$ are complex valued,
symmetric in the components of $z_{b} \in \vec\Lambdabold_b^*$, and
antisymmetric in the components of $z_{f} \in \vec\Lambdabold_f^*$.
Coefficients that obey these symmetry conditions are said to be
\emph{admissible}. Let $\Fcal$ be the set of admissible coefficients.
As vector spaces, $\Hcal$ and $\Fcal$ are isomorphic by the map $A
\mapsto (F_z)_{z\in\vec\Lambdabold^*}$ implicitly defined by
\eqref{e:Arep}.

We use this isomorphism to transport the product
from $\Hcal$ to a product on $\Fcal$.  Let
\begin{equation}
    \label{e:eta}
    \eta^{z}
=
    \xi^{z_{b}}\psi^{z_{f}}
.
\end{equation}
For $F',F'' \in \Fcal$, we define $(F' \star
F'')$ to be the unique element of $\Fcal$ such that
\begin{equation}
\label{e:star-product}
    \sum_{z \in \vec\Lambdabold^{*}}
    \frac{1}{z!} (F'\star F'')_{z}\eta^{z}
=
    \left(
    \sum_{z' \in \vec\Lambdabold^{*}}
    \frac{1}{z'!} F'_{z'}\eta^{z'}
    \right)
    \left(
    \sum_{z'' \in \vec\Lambdabold^{*}}
    \frac{1}{z''!} F''_{z''}\eta^{z''}
    \right)
.
\end{equation}
The vector space isomorphism between $\Hcal$ and $\Fcal$ implies the
existence of
$F'\star F''$, and with the $\star$ product, $\Fcal$ becomes an
algebra isomorphic to $\Hcal$.

For a sequence $x= (x_{1},x_{2},\dots ,x_{p})$, we say that $(x',x'')$
are \emph{complementary with respect to} $x$ if $x'$ is a subsequence
of $x$ and $x''$ is the sequence obtained by removing $x'$ from
$x$. The pairs such that $x'$ or $x''$ is the empty sequence are
included. We denote by $S_{x}$ the set of all pairs $(x',x'')$ that
are complementary with respect to $x$.  There is an inverse relation:
given sequences $x'$ and $x''$ we define $x'\diamond x''$ to be the
set of all $x$ such that $(x',x'') \in S_{x}$. We extend this notation
to $z \in \vec\Lambdabold^{*}$ by applying it to $z$ species by
species.  For example, with just one boson and one fermion species,
$(z',z'')$ are complementary with respect to $z$ if $(z'_{b},z''_{b})
\in S_{z_{b}}$ and $(z'_{f},z''_{f}) \in S_{z_{f}}$. We define $S_{z}$
to be the set of all $(z',z'')$ that are complementary with respect to
$z$ and we define $z'\diamond z''$ to be the set of all
$z\in\vec\Lambdabold^{*}$ such that $(z',z'') \in S_{z}$.  Recall that
factorials and concatenation are defined species-wise in
$\vec\Lambdabold^*$, in Section~\ref{sec:seq}.  Finally, given a
sequence $z$ with complementary subsequences $z'$ and $z''$, we define
$\sgn (z',z'';z)\in \{-1,1\}$ by the requirement that $\eta^{z} = \sgn
(z',z'';z)\eta^{z'}\eta^{z''}$.

\begin{lemma}
\label{lem:product-formula}
For $F',F'' \in \Fcal$,
the product defined on $\Fcal$ by
\eqref{e:star-product} is given by
\begin{equation}
\label{e:starproddef}
    (F' \star F'')_{z}
=
    \sum_{(z',z'') \in S_{z}}
    F'_{z'} F''_{z''} \,
    \sgn (z',z'';z)
.
\end{equation}
\end{lemma}

\begin{proof}
Let $(F' * F'')_{z}$ denote the right-hand side of
\eqref{e:starproddef}.
It suffices to show that
\eq
\label{e:closed}
     F'*F'' \in \Fcal,
\en
and
\eq
\label{e:rightprod}
    \sum_z \frac{1}{z!} (F'*F'')_z \eta^z
    =
    \text{right-hand side of \eqref{e:star-product}}.
\en

First, by definition,
\begin{align}
    \label{e:bullet-prod}
    (F' * F'')_{z}\eta^{z}
&=
    \sum_{(z',z'') \in S_{z}}
    F'_{z'} F''_{z''} \,
    \sgn (z',z'';z)\eta^{z}
=
    \sum_{(z',z'') \in S_{z}}
    F'_{z'} \eta^{z'} F''_{z''} \eta^{z''}
.
\end{align}
Therefore,
\begin{align}
    \sum_z \frac{1}{z!} (F'*F'')_z \eta^z
&=
    \sum_{z}
    \frac{1}{z!}
    \sum_{(z',z'')\in S_{z}}
    F'_{z'}\eta^{z'} F''_{z''}\eta^{z''}
=
    \sum_{z',z''}
    F'_{z'}\eta^{z'} F''_{z''} \eta^{z''}
    \sum_{z \in z'\diamond z''} \frac{1}{z!}.
\end{align}
The number of $z$ in the set $z'\diamond z''$ is
$z!/(z'!z''!)$, because each $z$ is specified by choosing a
subsequence $(j_{1},..,j_{p'})$ of $(1,\dots ,p (z))$ and setting
$z_{j_{k}}=z'_{k}$, with the other components of $z$ then determined
by $z''$.  This gives
\begin{align}
    \sum_z \frac{1}{z!} (F'*F'')_z \eta^z
&=
    \sum_{z',z''}
    \frac{1}{z'!}\frac{1}{z''!}
    F'_{z'} \eta^{z'} F''_{z''} \eta^{z''}
,
\end{align}
which proves \refeq{rightprod}.

For $F:\vec\Lambdabold^{*}\to\C$, let $\tilde{F}_{z} =
F_{z}\eta^{z}$.  The admissibility requirement in the definition
of $\Fcal$ is equivalent to the statement that $F\in \Fcal$ if and only if
$\tilde F_{\pi z} = \tilde F_z$ for any permutation  $\pi$
of $z$.
Also, given
$(z',z'') \in S_{\pi z}$, we can define
$(\hat z', \hat z'') \in S_{z}$ in a unique way by reordering
the components of $z'$ to produce $\hat z'$ and similarly for $z''$.
Then, by \refeq{bullet-prod},
\begin{align}
    \widetilde {(F' * F'')}_{\pi z}
&=
    \sum_{(z',z'') \in S_{\pi z}}
    \tilde{F}'_{z'}  \tilde{F}''_{z''}
    =    \sum_{(z',z'') \in S_{\pi z}}
    \tilde{F}'_{\hat z'}  \tilde{F}''_{\hat z''}
    =    \sum_{(z',z'') \in S_{z}}
    \tilde{F}'_{z'}  \tilde{F}''_{ z''}
    =
    \widetilde {(F' * F'')}_{ z}
    ,
\end{align}
where the second
equality holds since $F',F''\in\Fcal$.  This proves \refeq{closed}, and completes
the proof.
\end{proof}

Given $F \in \Fcal$ and a test function $g\in\Phi$, we define
a pairing and a semi-norm by
\begin{equation}
\label{e:Tpairdef}
    \pair{ F, g }
=
    \sum_{z \in \vec\Lambdabold^{*}}
    \frac{1}{z!} F_{z} g_{z}
,
\quad\quad
    \|F\|_{T}
=
    \sup_{g \in B (\Phi)} \,
    |\langle F , g \rangle |
.
\end{equation}
The following proposition shows that the $T$ semi-norm on $\Fcal$ obeys
the product property.

\begin{prop}
\label{prop:prodT}
For all $F,G \in \Fcal$,
$\|F\star G\|_{T} \leq \|F\|_{T}\|G\|_{T}$.
\end{prop}

\begin{proof}
Let $g\in\Phi $ and $G \in \Fcal$.
By Lemma~\ref{lem:product-formula},
\begin{align}
    \langle F\star G, g \rangle
&=
    \sum_z \frac{1}{z!} \sum_{(z',z'') \in S_z} F_{z'}G_{z''}
    \,\sgn (z',z'';z) \,g_z
\nnb
&=
    \sum_{z',z''}
    \sum_{z\in z' \diamond z''} \frac{1}{z!} F_{z'}G_{z''}
    \,\sgn (z',z'';z) \,g_z
.
\end{align}
We define $G^{*}g \in \Phi$ by
\begin{equation}
    (G^{*}g)_{z'}
    =
    \sum_{z''} \frac{1}{z''!} G_{z''}
    \sum_{z \in z' \diamond z''} \frac{z'!z''!}{z!}
    \,\sgn (z',z'';z) \, g_z,
\end{equation}
so that
\begin{align}
\lbeq{Gadj}
    \langle F\star G, g \rangle
&=
    \sum_{z'} \frac{1}{z'!} F_{z'}(G^{*} g)_{z'}
=
    \pair{F,G^{*}g}
\end{align}
and hence
\begin{equation}
    \|F\star G\|_{T}
    \le
    \|F\|_{T} \sup_{g \in B (\Phi)} \|G^{*}g\|_{\Phi}
.
\end{equation}
Thus it remains to show that
\eq
\label{e:Gstarbd}
   \|G^{*}g\|_{\Phi} \le \|G\|_{T} \quad \quad \text{for $g \in B (\Phi)$}.
\en

Given $g \in B(\Phi)$ and $z' \in \vec\Lambdabold^{*}$, we define a
test function $f_{z'} \in \Phi$ by setting its value $(f_{z'})_{z''}$ at
$z''$ to be equal to
\begin{equation}
\label{e:fzzdef}
    f_{z',z''}
=
    \sum_{z \in z' \diamond z''} \frac{z'!z''!}{z!}
    \,\sgn (z',z'';z) \,
    g_z,
\end{equation}
where $f_{z',z''}$ is a short notation for $(f_{z'})_{z''}$.
We regard this as a function of $z''$ with $z'$ fixed.  By definition,
$(G^{*}g)_{z'} = \pair{G,f_{z'}}$, and hence, by
Definition~\ref{def:gnorm-general},
\begin{equation}
    \|G^*g\|_\Phi
    =
    \sup_{(\alpha',z')\in \Acal'}
    |\lambda_{\alpha',z'} \pair{G,f_{z'}}|
    =
    \sup_{(\alpha',z')\in \Acal'}
    |\pair{G,\lambda_{\alpha',z'} f_{z'}}|,
\end{equation}
where $\Acal'$ denotes a copy of $\Acal$, and where we
have made the abbreviation
$\lambda_{\alpha',z'}= w_{\alpha',z'}^{-1}\nabla^{\alpha'}$.  Thus
we obtain
\begin{equation}
    \|G^*g\|_\Phi
    \le
    \|G\|_T
    \sup_{(\alpha',z')\in \Acal'}
    \|\lambda_{\alpha',z'} f_{z'}\|_\Phi.
\end{equation}
Thus, it is sufficient to show that for all $g \in B(\Phi)$
and $(\alpha',z')\in \Acal'$, $(\alpha'',z'') \in \Acal''$,
\begin{equation}
    |\lambda_{\alpha'',z''}\lambda_{\alpha',z'} f_{z',z''}| \le 1.
\end{equation}
In \refeq{fzzdef}, the operations $\lambda_{\alpha'',z''}\lambda_{\alpha',z'}$
can be interchanged with the summation because they are linear,
and with the factorials and sgn function
since these depend only on the length and order of the relevant sequences.
Since the number of terms in the sum over $z\in z' \diamond z''$ is equal to
$z!/(z'!z''!)$, we find after taking the absolute values inside the summation
that it suffices to show that, for each $z \in z'\diamond z''$,
\begin{equation}
\label{e:lam1}
    |\lambda_{\alpha'',z''}\lambda_{\alpha',z'} g_z| \le 1,
\end{equation}
where the derivatives within the $\lambda$ factors act on the
arguments of $g_z$ according to their permuted locations within
$z \in z'\diamond z''$.  Since \refeq{lam1} is a consequence
of $g \in B(\Phi)$ and the definition of the $\Phi$ norm, this
completes the proof.
\end{proof}

\begin{proof}[Proof of Proposition~\ref{prop:prod}.]  Let
$F =\sum_{y \in \Lambda_f^*}\frac{1}{y!}F_y \psi^y \in \Ncal$.
For boson fields $\phi,\xi$, Taylor expansion of the
coefficients $F_{y}$ about a fixed $\phi$ in powers of $\xi$
defines an algebra isomorphism
\begin{equation}
    F
\mapsto
    \sum_{(x,y) \in \vec\Lambdabold^{*}}
    \frac{1}{x!y!} F_{x,y} (\phi) \xi^{x}\psi^y
\end{equation}
of $\Ncal$ into a subalgebra (if $p_\Ncal < \infty$)
of the algebra $\Hcal$ and, in turn, $\Hcal$ is
isomorphic as an algebra to $\Fcal$.  The composition of these
isomorphisms is an isometry of the semi-normed algebras
$(\Ncal,T_{\phi})$ and $(\Fcal,T)$, so
Proposition~\ref{prop:prod} follows from Proposition~\ref{prop:prodT}.
\end{proof}

Finally, we extract and develop a detail from the proof of
Proposition~\ref{prop:prodT}, needed only in
\cite[Section~\ref{loc-sec:LTnormestimates}]{BS-rg-loc}.  Examination
of the proof of \refeq{Gadj} shows that it is also true that
$\pair{F\star G,g} = \pair{G,F^\dagger g}$ for all $F,G\in \Fcal$ and
$g \in \Phi$, where
\begin{align}
\label{e:Fdag}
    (F^{\dagger}g)_{z''}
    &=
    \sum_{z'} \frac{1}{z'!}F_{z'}
    \sum_{z \in z' \diamond z''} \frac{z'!z''!}{z!}
    \,\sgn(z',z'';z) \, g_z
.
\end{align}
By the isomorphism mentioned in the proof of
Proposition~\ref{prop:prod}, \refeq{Fdag} also defines an adjoint in
$\Ncal$, in the sense that $\pair{FG,g}_\phi = \pair{G,F^\dagger
g}_\phi$ also for $F,G\in \Ncal$.  We apply this to the case of a test
function $f_z$ which is nonzero only on sequences $z$ of fixed length
$p(z)=n$ and of fixed choice of species for each of the $n$ components
of $z$.  In this case, $z!=n!$, and given $z'$ in the sum in
\refeq{Fdag}, $z''!$ is determined by $z'!$ (and by the fixed value of
$n$).  In addition, given $z'$, it is also the case that
$\sgn(z',z'';z)$ is determined since the species in $z$ are known when
$f_z \neq 0$.  Thus there are coefficients $c_{z'}=
\frac{z''!}{z!}\sgn(z',z'';z)$ such that, for the special $f$ under
consideration,
\begin{align}
\label{e:Fdagf}
    (F^{\dagger}f)_{z''}
    &=
    \sum_{z'} c_{z'}F_{z'} \tilde f^{(z')}_{z''}
    \quad \text{with} \quad
    \tilde f^{(z')}_{z''} =     \sum_{z \in z' \diamond z''} f_z
.
\end{align}

\subsection{Exponential norm estimate}
\label{sec:ebdTphi}

In this section, we prove Proposition~\ref{prop:eK}.

Let $f(u)=\sum_{n=0}^\infty a_n u^n$
and $h(u)=\sum_{n=0}^\infty  |a_n| u^n$,
and let $\|\cdot\|$ denote any semi-norm that obeys the
product property, e.g., the $T_\phi$ semi-norm. As an immediate consequence
of the product property, for any $F$, we have
\eq
\label{e:apos}
    \|f(F)\| \leq \sum_{n=0}^\infty |a_n| \|F^n\|
    \leq \sum_{n=0}^\infty |a_n| \|F\|^n = h(\|F\|).
\en
It follows from \refeq{apos} that
\eq
\label{e:apos-e}
    \|e^{-F}\|_{T_\phi} \le e^{\|F \|_{T_\phi}}.
\en
Proposition~\ref{prop:eK} provides an
improvement to \refeq{apos-e} when
the purely bosonic part of
$F$ has positive real part.  Its proof is based
on the following lemma.

\begin{lemma}
\label{lem:fK}
Let $\|\cdot\|$ denote any semi-norm that obeys the product property.
If $\|F\| \le 1$, then
\begin{align}
    \label{e:RK3}
    \|e^{-\frac{1}{2}F^{2}}-(1+F)e^{-F}\|
    &
    \le
    \|e^{-\frac{1}{2}F^{2}}\|\,\|F^{3}\|
    .
\end{align}
\end{lemma}

\begin{proof}
Let
\eq
\label{e:fK7}
    R
=
    e^{-\frac{1}{2}F^{2}}-(1+F)e^{-F}.
\en
Let
\eq
    f (z)
=
    1+(z-1)e^{z+\frac{1}{2}z^{2}}.
\en
Then $R=e^{-\frac 12 F^2}f(-F)$.
By definition, $f (0)=0$ and
$
    f' (z)
=
    z^{2}
    e^{z+\frac{1}{2}z^{2}}.
$
Thus $f'(z)$ has a power series with non-negative
coefficients, and hence so does $f(z)$.  Also, $f (z)= z^{3}g
(z)$, for some $g (z)=\sum b_{n}z^{n}$ with $b_{n}\ge 0$.   In
addition, $g(1)=f(1)=1$. Therefore, by \refeq{apos},
\eq
    \|f (-F) \|
=
    \|F^{3}g(-F)\|
\le
    \|F^{3}\|\|g(-F) \|
    \leq
    \|F^{3}\|\,g(\|F\|) .
\en
If $\|F\|\le 1$, this simplifies to
\eq
\label{e:f-K}
    \|f (-F) \|
\le
    \|F^{3}\|.
\en
This gives
\eq
\label{e:fK1}
    \|R\|
=
    \big\|
    e^{-\frac{1}{2}F^{2}} f (-F)
    \big\|
\le
    \|e^{-\frac{1}{2}F^{2}}\|
    \|F^{3}\|,
\en
which is \refeq{RK3}.
\end{proof}

\begin{proof}[Proof of Proposition~\ref{prop:eK}.]
Let $F\in \Ncal$ and let $F_\varnothing(\phi)$ be the
purely bosonic part of $F$.  We will prove that
\eq
    \|e^{-F  }\|_{T_\phi}
    \le
    e^{-2{\rm Re} F_\varnothing(\phi)+\|F  \|_{T_\phi}}
    .
\en

We first assume that $|F_\varnothing(\phi)|$ is sufficiently small
that $|1-F_\varnothing(\phi)| - |F_\varnothing(\phi)| \ge 0$,
and we write $F_\varnothing(\phi)=z=x+iy$.
We show that this implies that
\eq
\label{e:eKF0}
    |1-z| - |z| \le 1 - 2x,
\en
as follows.
By hypothesis, \refeq{eKF0} is equivalent to the
inequality obtained by squaring both sides, and algebra reduces the
latter to
\eq
    x(1-x)+y^2 \le |z| \, |1-z|.
\en
This certainly holds if the left-hand side is negative, and otherwise
it suffices to show that the inequality is valid if both sides are
squared, and the latter reduces to
\eq
    2x(1-x) \le (1-x)^2+x^2,
\en
which does hold.  This completes the proof of \refeq{eKF0} when
$|1-F_\varnothing(\phi)| - |F_\varnothing(\phi)| \ge 0$.

The $T_\phi$ semi-norm is defined via the pairing given in
\eqref{e:Kgpairdef}.  Let $g$ be any test
function of norm at most $1$.
By separating out the null contribution to the sum
over $z$ we have
\begin{align}
    |\langle 1-F , g \rangle_\phi|
    &\le
    |(1-F_\varnothing(\phi))g_\varnothing| +
    \sum_{r\not = 0}
    \big|
    \sum_{z\in \vec\Lambdabold^{(r)}}
    \frac{1}{z!} F_{z}(\phi) g_{z}
    \big|
    \nnb
    &=
    (|1-F_\varnothing(\phi)| -|F_{\varnothing}(\phi)|)|g_\varnothing | +
    \sum_{r \ge 0}
    \big|
    \sum_{z\in \vec\Lambdabold^{(r)}}
    \frac{1}{z!} F_{z}(\phi) g_{z}
    \big|.
\end{align}
We take the supremum over test functions $g$ of unit norm.  The final
term becomes $\|F\|_{T_\phi}$, so
\begin{equation}
    \label{e:explinbis}
    \|1-F \|_{T_\phi}
    \le
    |1-F_\varnothing(\phi)| -|F_{\varnothing}(\phi)| + \|F\|_{T_\phi}.
\end{equation}

For the rest of the proof, we drop the $T_{\phi}$ subscript.
Given $\phi$, we choose $N$ sufficiently large that
$|1-\frac{1}{N} F_\varnothing(\phi)| - \frac{1}{N}|F_\varnothing(\phi)| \ge 0$.
By \refeq{explinbis} and \refeq{eKF0},
\eq
    \left\|1-\frac{1}{N}F \right\| \le
    1 - \frac{2}{N}{\rm Re} F_\varnothing(\phi) +
    \frac{1}{N}\|F\|.
\en
By the
product property,
\eq
    \left\|\big(1-\frac{1}{N}F \big)^{N}\right\| \le
    \big(
    1 - \frac{2}{N}{\rm Re} F_{\varnothing}(\phi) +
    \frac{1}{N}\|F\|
    \big)^{N} \le
    e^{-2{\rm Re} F_\varnothing(\phi)+\|F\|}.
\en
It suffices now to show that the limit $N \to \infty$ can be taken
inside the semi-norm on the left-hand side.  For this we define $A =
e^{F/N}(1-\frac{1}{N}F)$.
By \eqref{e:apos} with $f(z)=e^z$ and with $f (z) = e^{z}-1$, we have
$\|e^{-\frac 12 F^2/N^2}\| = O(1)$ and
$\|e^{-\frac 12 F^2/N^2}-1\|=O(N^{-2})$ as $N \to \infty$.  Therefore,
by \eqref{e:RK3} with $F$ replaced by
$-F/N$,
\begin{align}
    \|A-1\|
    &\le
    \|e^{F/N}(1-\frac{1}{N}F)- e^{-\frac{1}{2}F^{2}/N^{2}}\|
    +
    \|e^{-\frac 12 F^2/N^2}-1\|
    \nnb &
    =
    O (N^{-3}) + O(N^{-2})
    =
    O (N^{-2}).
\end{align}
Now let $f (z) = (1-z^{N})
(1-z)^{-1} = \sum_{n=0}^{N-1} z^{n}$.  Then, by \eqref{e:apos},
\begin{align}
    \big\| \big(1-\frac{1}{N}F \big)^{N} - e^{-F} \big\|
    &=
    \|e^{-F}
    (
    A^{N} - 1
    )
    \|
    \nnb &
    \le
    \|e^{-F}\|\,
    \|A-1\|\, f (\|A\|)
    =O(N^{-2})f(\|A\|),
\end{align}
and the right-hand side is $O(N^{-1})$ since $f(1+O(N^{-2}))=O(N)$.
This completes the proof.
\end{proof}

\subsection{Polynomial norm estimate}
\label{sec-Tphiestimates}

In this section, we prove Proposition~\ref{prop:T0K}.  We begin with
some definitions and a preliminary lemma which will be useful also in
Sections~\ref{sec:normchange}--\ref{sec:compsp0}.

For $z \in \vec\Lambdabold^{*}$, let $B_z$ denote the set of pairs
$(z',z_b'')\in \vec\Lambdabold^{*}\times \vec\Lambdabold_b^{*}$ such
that $z'\concat z_b''=z$.  For $s \in [0,1]$, $g\in\Phi$, $\xi
\in \R^{\Lambdabold_b}$ and $z \in \vec\Lambdabold^{*}$, we define a
new test function $\sigma^{*} (s) g \in \Phi$ by setting $(\sigma^{*}
(s) g)_{z} = 0$ if the length of $z$ exceeds $p_\Ncal$, and otherwise
\begin{equation}
\label{e:sigmastardef}
    (\sigma^{*}_\xi (s) g)_{z}
=
    \sum_{(z',z_b'')\in B_z} \frac{z!}{z'!z_b''!}
    s^{z_b''} \xi^{z_b''} g_{z'}
.
\end{equation}
We write $\sigma^{* (m)}_\xi g$ to denote
the $m^{\rm th}$ derivative of $\sigma_\xi^* (s) g$ at $s=0$.

\begin{lemma}
\label{lem:Stm} For $s,t \in [0,1]$, $g\in\Phi$, $P \in \Ncal$ a
polynomial of degree at most $p_\Ncal$, and $\phi,\xi \in
\R^{\Lambdabold_b}$,
\begin{equation}
\label{e:sigstar}
    \pair{P,g}_{t\phi + s\xi}
=
    \pair{P,\sigma_\xi^{*} (s) g}_{t\phi}
.
\end{equation}
If $g \in \Phi^{(p)}$ and $m+p \le p_\Ncal$, then $\sigma_\xi^{*(m)}g\in
\Phi^{(m+p)}$, and, for any $F\in \Ncal$,
\begin{equation}
\label{e:sigstarmid}
    \frac{d^m}{dt^m}
    \pair{F,g}_{t\phi}
    =
    \pair{ F, \sigma_\phi^{*(m)}g }_{t\phi}.
\end{equation}
For all $p$ and for $g \in \Phi $,
\begin{equation}
\label{e:sigstar0}
    \|\sigma^{*}_\xi (1) g\|_{\Phi^{(p)}}
\le
    \left(1 + \|\xi\|_{\Phi}\right)^{p}
    \|g\|_{\Phi}
,
\end{equation}
and
\begin{equation}
\label{e:sigstarm}
    \| \sigma_\xi^{*(m)} g\|_{\Phi^{(m+p)}}
\le
    \frac{(m+p)!}{p!}\|\xi\|_{\Phi}^{m}
    \|g\|_{\Phi}
.
\end{equation}
\end{lemma}

\begin{proof}
By definition, for $g\in\Phi$
and for a polynomial $P$ of degree $p_\Ncal$,
\begin{align}
    \pair{P,g}_{t\phi + s\xi}
&=
    \sum_{z'} \frac{1}{z'!} P_{z'} \left(t\phi + s\xi \right)
    g_{z'}
=
    \sum_{z',z_b''} \frac{1}{z'!z_b''!} P_{z'\concat z_b''} (t\phi)
    s^{z_b''} \xi^{z_b''} g_{z'}
\nnb
&=
    \sum_{z}\frac{1}{z!} P_{z} (t\phi)
    \sum_{( z',z_b'') \in B (z)} \frac{z!}{z'!z_b''!}
    s^{z_b''} \xi^{z_b''} g_{z'}
=
    \pair{P,\sigma^{*}_\xi (s) g}_{t\phi}
,
\end{align}
which proves \eqref{e:sigstar}.  If $g \in \Phi^{(p)}$ then
differentiation of \eqref{e:sigmastardef} gives
\begin{equation}
\label{e:sigmastardiff}
    (\sigma^{*(m)}_\xi  g)_{z}
=
    \1_{z'\concat z_b''=z} \frac{(m+p)!}{p!}
    \xi^{z_b''} g_{z'}
,
\end{equation}
so $\sigma^{*(m)}_\xi g\in \Phi^{(m+p)}$.  Also, when $g \in \Phi^{(p)}$,
we may regard $F$ in \refeq{sigstarmid} as a polynomial and thus by
\refeq{sigstar} we obtain \eqref{e:sigstarmid} via differentiation
with respect to $s$ (with $\xi=\phi$).

By the triangle inequality and \eqref{e:plusnormass}, for $p \le
p_{\Ncal}$,
\begin{equation}
    \|\sigma^{*}_\xi (1) g\|_{\Phi^{(p)}}
\le
    \sumtwo{p',p_b'':}{p'+p_b''=p}
    \frac{p!}{p'! p_b''!}
    \|\xi\|_{\Phi}^{p_b''}
    \|g\|_{\Phi}
=
    \left(1 + \|\xi\|_{\Phi}\right)^{p}
    \|g\|_{\Phi}
.
\end{equation}
Since the left-hand side is zero for $p>p_{\Ncal}$, this proves
\eqref{e:sigstar0}.  For \eqref{e:sigstarm}, we only consider
the case $m+p \le p_\Ncal$ because otherwise the left-hand side is
zero. Also we can assume that $g \in\Phi^{(p)}$ because no other part
of $g$ can contribute to the left-hand side.  If $g \in\Phi^{(p)}$
and $m+p \le p_\Ncal$, then from \eqref{e:sigmastardiff} and
\eqref{e:plusnormass} we have
\begin{equation}
    \| \sigma^{*(m)}_\xi g\|_{\Phi^{(m+p)}}
\le
    \frac{(m+p)!}{p!}
    \|\xi\|_{\Phi}^{m}
    \| g \|_{\Phi}
.
\end{equation}
This proves \eqref{e:sigstarm}, and completes the proof.
\end{proof}

\begin{rk}
\label{rk:pair-deriv}
It follows from
Lemma~\ref{lem:Stm} that for $F \in \Ncal$, $g \in \Phi^{(p)}$,
and for $m+p \le p_\Ncal$,
\begin{equation}
\label{e:shift-phi}
    \left|
    \frac{d^{m}}{dt^{m}}
    \pair{F,g}_{T_{\phi+t\xi} }
    \right|
\le
    \frac{(m+p)!}{m!}
    \|F\|_{T_{\phi+t\xi} (\Phi)}
    \|\xi\|_{\Phi}^{m}
    \|g\|_{\Phi}
.
\end{equation}
To see this, note that as in the proof of Lemma~\ref{lem:Stm},
\eq
    \frac{d^{m}}{ds^{m}} \Big|_{0}
    \pair{F,g}_{T_{\phi+s\xi} (\Phi)}
    =
    \langle{F,\sigma_\xi^{*(m)}g\rangle}_{t\phi}.
\en
With \eqref{e:sigstarm}, this gives
\begin{align}
    \left|
    \frac{d^{m}}{ds^{m}} \Big|_{0}
    \pair{F,g}_{T_{\phi+s\xi} (\Phi)}
    \right|
&\le
    \frac{(m+p)!}{m!}
    \|F\|_{T_{\phi} }
    \|\xi\|_{\Phi}^{m}
    \|g\|_{\Phi}
,
\end{align}
and then \refeq{shift-phi} follows
by replacing $\phi$ with $\phi + t \xi$.
\end{rk}

For $F \in \Ncal$ and $t \ge 0$, we define $\tau_t F \in \Ncal$ by
replacing the fields $(\phi ,\psi)$ in $F$ by $(t\phi ,t\psi)$.  For a
positive integer $A$, $t \in [0,1]$ and $F \in \Ncal$, we define the
truncated Taylor expansion for $\tau_{t}F$ by
\begin{equation}
\label{e:sigmaA}
    \tau_{t}^{(\le A)}F
=
    \sum_{n=0}^{A}\frac{t^{n}}{n!}\tau^{(n)}_{0}F
,
\end{equation}
where $\tau_{t}^{(n)}F$ is the $n^{\rm th}$ derivative of $\tau_{t}F$
with respect to $t$.
The following lemma gives the result of Proposition~\ref{prop:T0K}.

\begin{lemma}
\label{lem:Tphi-pol-bound} For $F \in \Ncal$ and $A \le p_\Ncal$,
let $P=\tau_{1}^{(\le A)}F$. Then
\begin{align}
&
    \label{e:SsleA}
    \|P\|_{T_{\phi}}
\le
    \|F\|_{T_{0}}
    \big(1+ \|\phi\|_{\Phi}\big)^{A}
    ,
\end{align}
and if $F$ is a polynomial of degree $A$ then
\begin{equation}
\label{e:Fpoly}
    \|F\|_{T_{\phi}}
\le
     \|F\|_{T_{0}}\big(1+ \|\phi\|_{\Phi}\big)^{A}
.
\end{equation}
\end{lemma}

\begin{proof}
The second claim is a consequence of the first because, in this case,
$F=\tau_{1}^{(\le A)}F$ by the uniqueness of Taylor expansions.  To
prove \eqref{e:SsleA}, we apply Lemma~\ref{lem:Stm} with $\xi
=\phi$, $t=0$ and $s=1$ to obtain
\begin{equation}
    \left|\pair{P ,g}_{\phi}\right|
=
    \left|\pair{P,\sigma^{*}_\phi (1) g}_{0}\right|
\le
   \|P\|_{T_{0}}\,\|\sigma^{*}_\phi (1) g\|_{\Phi}
.
\end{equation}
Since $\|P\|_{T_{0}}$ is a truncation of the sum of positive terms
that constitute $\|F\|_{T_{0}}$, it is the case that $\|P\|_{T_{0}}
\le \|F\|_{T_{0}}$.  Also, we need only consider the case where
$\sigma^{*}_\phi (1) g$ depends on at most $A$ variables, since otherwise
its pairing with $P$ vanishes.  It then follows from
Lemma~\ref{lem:Stm} with $\xi =\phi$ that
\begin{equation}
    \left|\pair{P ,g}_{\phi}\right|
\le
    \|F\|_{T_{0}}\,\left(1 + \|\phi\|_{\Phi}\right)^{A}
    \|g\|_{\Phi}
.
\end{equation}
Taking the supremum now over $g\in B(\Phi)$, we obtain \eqref{e:SsleA}
and the proof is complete.
\end{proof}

\subsection{Estimate with change of norm}
\label{sec:normchange}

The following lemma gives the result of
Proposition~\ref{prop:Tphi-bound}.

\begin{lemma}
\label{lem:Tphi-bound-bis} Let $A<p_\Ncal$ be a non-negative integer.
For $F \in\Ncal$, let $P=\tau_{1}^{(\le A)}F$.  Then
\begin{align}
    \|F\|_{T_{\phi}'}
    & \leq
    \left(1 + \|\phi\|_{\Phi'}\right)^{A+1}
    \left( \|P\|_{T_{0}'}
    +
    \rho^{(A+1)} \sup_{t \in [0,1]}
    \|F\|_{T_{t\phi}}  \right)
    \nnb
    & \leq
    \left(1 + \|\phi\|_{\Phi'}\right)^{A+1}
    \left( \|F\|_{T_{0}'}
    +
    \rho^{(A+1)} \sup_{t \in [0,1]}
    \|F\|_{T_{t\phi}}  \right).
\label{e:Tphicor}
\end{align}
\end{lemma}

\begin{proof} The second estimate follows from the first and
$\|P\|_{T_{0}'} \le \|F\|_{T_{0}'}$.  To prove the first inequality
let $R= F - P$. By the triangle inequality and
Lemma~\ref{lem:Tphi-pol-bound} it is sufficient to prove that
\begin{align}
\label{e:tautmGm}
    \|R\|_{T'_{\phi}}
\le
    \rho^{(A+1)}\,
    \left(1 + \|\phi\|_{\Phi'}\right)^{A+1}
    \sup_{t \in [0,1]}
    \|F\|_{T_{t \phi}}
.
\end{align}
For this, it suffices to show that for a test function $g \in\Phi$ we
have
\begin{align}
\label{e:tautmGmpair}
    |\pair{R,g}_\phi|
\le
    \rho^{(A+1)}\,
    \left(1 + \|\phi\|_{\Phi'}\right)^{A+1}
    \sup_{t \in [0,1]}
    \|F\|_{T_{t \phi}} \|g\|_{\Phi'}
.
\end{align}
We consider separately
the cases (i) $g_{z}=0$ for $z$ with $p=p (z) \le A$, and
(ii) $g_{z}=0$ except when $p=p(z) = 0,1,\dots ,A$.  Any $g$
can be decomposed into these two cases using
$g = g\1_{p > A}+ \sum_{r \le A}g\1_{p =r}$, and
\begin{equation}
    |
    \pair{R,g}_{\phi}
    |
\le
    |\pair{R,g\1_{p > A}}_{\phi}|
    +
    \sum_{r \le A}|\pair{R,g\1_{p =r}}_{\phi}|.
\end{equation}

For case (i), we simply note from \eqref{e:rhodef1} that
\begin{align}
    \left|
    \pair{R,g}_{\phi}
    \right|
&=
    \left|
    \pair{F,g}_{\phi}
    \right|
\le
    \|F\|_{T_{\phi}}\|g\|_{\Phi}
\le
    \|F\|_{T_{\phi}}\frac 12 \rho^{(A+1)}
    \|g\|_{\Phi'}
.
\end{align}
Note that the above right-hand side is at most half the right-hand
side of \eqref{e:tautmGmpair}.

For the more substantial case (ii), fix $g\in \Phi$ with $g_{z}=0$
supported on sequences of length exactly $p$ with some $p\in
\{0,1,\dots ,A \}$.  Let $f (t) = \pair{R,g}_{t\phi}$. By the Taylor
remainder formula, for any $m \le A+1$,
\begin{equation}
    \left| \pair{R,g}_{\phi}\right|
\le
    \frac{1}{m!} \sup_{t \in [0,1]} \left| f^{(m)} (t)  \right|
.
\end{equation}
By Lemma~\ref{lem:Stm} with $\xi =\phi$,
\begin{equation}
\label{e:fRg}
   f^{(m)} (t)
=
   \pair{R,\sigma_\phi^{*(m)}g}_{t\phi}
.
\end{equation}
Let $m=A+1-p$. Then we can
replace $R=F-P$ by $F$ in \eqref{e:fRg} because $\sigma_\phi^{*(m)}g$ is
supported on sequences of length $m+p=A+1$
by Lemma~\ref{lem:Stm}, whereas $P$ is a polynomial of
degree $A$. Therefore,
\begin{equation}
\label{e:ii-1}
    \left| \pair{R,g}_{\phi} \right|
\le
    \frac{1}{m!}\sup_{t \in [0,1]}
    \|F\|_{T_{t\phi}}
    \;
    \| \sigma_\phi^{*(m)}g\|_{\Phi^{(A+1)}}
.
\end{equation}
Since $\sigma_\phi^{*(m)}g$ is supported on sequences of length $A+1$, by
\eqref{e:rhodef1} we have
\begin{equation}
\label{e:ii-2}
    \| \sigma_\phi^{*(m)}g\|_{\Phi^{(A+1)}}
\le
    \frac 12 \rho^{(A+1)}\| \sigma_\phi^{*(m)}g\|_{\Phi'^{(A+1)}}
.
\end{equation}
It follows from \eqref{e:ii-1}--\eqref{e:ii-2}
and Lemma~\ref{lem:Stm} that
\begin{align}
    \left| \pair{R,g}_{\phi} \right|
&\le
    \frac 12 \rho^{(A+1)}
    \sum_{p=0}^A
    \binom{A+1}{p}
     \|\phi\|_{\Phi'}^{A+1-p}
     \|g\|_{\Phi'}
     \sup_{t \in [0,1]}
    \|F\|_{T_{t\phi}}
\nnb &\le
    \frac 12 \rho^{(A+1)}
    \left(
    1+ \|\phi\|_{\Phi'} \right)^{A+1}
     \|g\|_{\Phi'}
     \sup_{t \in [0,1]}
    \|F\|_{T_{t\phi}}
.
\end{align}
Combined with the estimate for case (i), this gives
\eqref{e:tautmGmpair} and completes the proof.
\end{proof}

\subsection{Contractive bound on \texorpdfstring{$\theta$}{theta}}
\label{sec:compsp0}

In this section, we prove Proposition~\ref{prop:derivs-of-tau-bis}.

Recall from the discussion above \eqref{e:phi-extended} that there is
a bijection between a subset of $\Lambdabold$ and $\Lambdabold'$,
written $x \mapsto x'$.  Recall from the discussion above
Proposition~\ref{prop:derivs-of-tau-bis} that species in $\Lambdabold$
and species in $\Lambdabold'$ are distinct, and are ordered in such a
way that a species from $\Lambdabold'$ occurs immediately following
its counterpart in $\Lambdabold$.  The \emph{forget} function
$f:\Lambdabold \sqcup \Lambdabold' \to \Lambdabold$ is defined by
setting $f (x')=x$ when $x' \in \Lambdabold'$ and $f (x)=x$ when $x
\in \Lambdabold$.  We extend $f$ to a map from
$(\overrightarrow{\Lambdabold \sqcup \Lambdabold'})^{*}$ to
$\vec\Lambdabold^{*}$ by letting $f$ act componentwise on sequences.
We define a map $\theta^{*}:
\Phi(\Lambdabold \sqcup \Lambdabold') \to \Phi (\Lambdabold)$ by
setting
\begin{equation}
\label{e:thetastar}
    (\theta^{*}g)_{z}
=
    \sumtwo{v \in (\overrightarrow{\Lambdabold \sqcup \Lambdabold'})^{*}:}
    {f (v) = z}
    \frac{z!}{v!}g_{v}.
\end{equation}
By definition, the $v!$ appearing in the above equation is equal
to $u!u'!$, where $u$ and $u'$ are respectively the subsequences of
$v$ drawn from $\Lambdabold$ and $\Lambdabold'$.

\begin{lemma}
\label{lem:thetaadj}
For $F \in \Ncal(\Lambdabold)$, $g\in \Phi(\Lambdabold \sqcup \Lambdabold')$,
$\phi \in \R^{\Lambdabold_b}$ and
$\xi \in \R^{\Lambdabold_b'}$,
\begin{align}
\label{e:theta-dual}
    \pair{\theta F,g}_{\phi \sqcup \xi}
&=
    \pair{F,\theta^{*}g}_{\phi + \xi}
.
\end{align}
\end{lemma}

\begin{proof}
First, we compute the coefficients $(\theta F)_v$ for
$v \in (\overrightarrow{\Lambdabold \sqcup \Lambdabold'})^{*}$,
which is what is relevant for the pairing of $\theta F$ with $g$.
By Definition~\ref{def:theta-new},
\begin{equation}
    \theta F
=
    \sum_{z_{f} \in \vec\Lambdabold_{f}^{*}}
    \frac{1}{z_{f}!}F_{z_{f}} (\phi+\xi) ( \psi + \psi')^{z_{f}}
.
\end{equation}
We expand $F_{z_{f}} ((\phi+\xi)+ (\hat\phi+\hat\xi))$ in
a power series in $\hat\phi+\hat\xi$ to obtain
\begin{equation}
    \theta F
=
    \sum_{z \in \vec\Lambdabold^{*}}
    \frac{1}{z!}F_{z} (\phi+\xi) (\hat\phi+\hat\xi)^{z_b} ( \psi + \psi')^{z_{f}}
.
\end{equation}
Now we expand the binomials on the right-hand side and reorder
the species within both the bosonic and fermionic products.
We reorder the subscript on $F_z$ in exactly the same way; then
no sign change occurs.  From this, we can read off the coefficients
\eq
    (\theta F)_v = F_{f(v)}(\phi+\xi).
\en
We abbreviate the right-hand side as $F_{f(v)}=F_{f(v)}(\phi+\xi)$.
Then
\begin{equation}
    \pair{\theta F,g}_{\phi \sqcup \xi}
    =
    \sum_{v \in (\overrightarrow{\Lambdabold \sqcup \Lambdabold'})^{*}}
    \frac{1}{v!}
    F_{f(v)} g_v
    =
    \sum_{z\in \vec\Lambdabold^{*}}
    \frac{1}{z!}
    F_{z}
    \sum_{v:f (v) = z}
    \frac{z!}{v!}
    g_{v}
    =
    \pair{F,\theta^{*}g}_{\phi + \xi},
\end{equation}
and the proof is complete.
\end{proof}

\begin{lemma}
\label{lem:thetastar} The map $\theta^{*} : \Phi(w\sqcup w') \to \Phi
(w+w')$ is a contraction, namely, for $g \in \Phi(w\sqcup w')$,
\begin{align}
\label{e:theta-bound11}
    \| \theta^{*}g \|_{\Phi(w+w')}
&\le
    \|g\|_{\Phi(w\sqcup w')}
    .
\end{align}
\end{lemma}

\begin{proof}
In the following,
$v\in(\overrightarrow{\Lambdabold \sqcup \Lambdabold'})^{*}$ and
$z\in\vec\Lambdabold^{*}$.
By \refeq{thetastar},
\begin{align}
    \label{e:theta-bound6}
    \left|
    (w + w')_{\alpha ,z}^{-1} (\nabla^{\alpha} \theta^{*}g)_{z}
    \right|
&\le
    (w + w')_{\alpha ,z}^{-1}
    \sum_{v: f (  v) = z}
    \frac{z!}{v!}
    \left| \nabla^{\alpha} g_{v} \right|
\nnb
&\le
    \|g\|_{\Phi(w\sqcup w')}(w + w')_{\alpha ,z}^{-1}
    \sum_{v: f (  v) = z}
    \frac{z!}{v!}
    (w\sqcup w')_{\alpha ,v}
    .
\end{align}
The final sum equals $(w + w')_{\alpha ,z}$ by the binomial
theorem; to see this we recall that $v$ has species segregated so
that in particular primed and unprimed variables are not interleaved, and
the binomial coefficient $z!/v!$ accounts for the number of ways
to desegregate these variables.
Then \refeq{theta-bound11} follows
by taking the supremum over
$(\alpha,z) \in \Acal$.
\end{proof}

\begin{proof}[Proof of Proposition~\ref{prop:derivs-of-tau-bis}]
Let $g \in B (\Phi (w\sqcup w'))$.
By Lemma~\ref{lem:thetaadj},
\begin{equation}
    \left|
    \pair{\theta F,g}_{\phi \sqcup \xi}
    \right|
    =
    \left|
    \pair{ F,\theta^* g}_{\phi + \xi}
    \right|
\le
    \|F \|_{T_{\phi + \xi}(w+w')}
    \|\theta^{*}g\|_{\Phi (w+ w')}
.
\end{equation}
Taking the supremum over $g \in B (\Phi (w\sqcup w'))$ and applying
Lemma~\ref{lem:thetastar}, we have
\begin{equation}
    \|\theta F\|_{T_{\phi \sqcup \xi} (w\sqcup w')}
\le
    \|F\|_{T_{\phi +\xi}(w+w')}
.
\end{equation}
This proves \eqref{e:theta-bd1}.
\end{proof}

\section{Integration norm estimates}
\label{sec:integration}

In this section, we prove Propositions~\ref{prop:Etau-bound},
\ref{prop:EK} and \ref{prop:EG2}.

\subsection{Laplacian norm estimates}
\label{sec:heat}

In this section, we prove Proposition~\ref{prop:Etau-bound}.
For this, it suffices to prove the following
lemma, which slightly improves \refeq{DCK} by reducing the
factor $A^2$ to $A(A-1)$ on its right-hand side.

\begin{lemma}
\label{lem:Etau-boundzz}
If $F \in \Ncal$ is a polynomial of degree at most $A$, with $A \le p_\Ncal$,
then
\begin{equation}
\label{e:DCKzz}
    \frac 12 \|\Delta_{\Cbf} F\|_{T_{\phi}}
    \le
    \binom{A}{2}\,\|\Cbf\|_{\Phi}\,\|F\|_{T_{\phi}}.
\end{equation}
\end{lemma}

\begin{proof}
For $g \in \Phi$ and $v \in \vec\Lambdabold^*$, let
\begin{equation}
\label{e:Cstarg}
    (\Cbf^{*}g)_{v}
=
    \1_{u\concat z=v}
    \frac{v!}{u!z!} \Cbf_{u} g_{z}
\end{equation}
if the length of $v$ is at most $A$, and otherwise $(\Cbf^{*}g)_{v}=0$.
Here $u$ denotes the first two coordinates of $v$ and $z$ denotes
the others; in particular $(\Cbf^{*}g)_{v} =0$ if the length of $v$ is
less than $2$.
Then, by the definition of the Laplacian in \refeq{LapC},
\begin{align}
    \frac{1}{2} \pair{\Delta_{\Cbf} F,g}_{\phi}
&=
    \frac{1}{2}
    \sum_{z} \frac{1}{z!} (\Delta_{\Cbf} F(\phi))_{z} g_{z}
=
    \sum_{u,z} \frac{1}{u!z!}  \Cbf_{u} F_{u \concat z}(\phi) g_{z}
\nnb &
=
    \sum_{v}\frac{1}{v!} F_{v}(\phi)
    (\Cbf^{*}g)_{v}
=
    \pair{F,\Cbf^{*}g}_{\phi}
.
\label{e:LapCpair}
\end{align}
Since $F$ is a polynomial of degree at most $A$, $F_v=0$ as soon
as the length of $v$ exceeds $A$;
the fact that $A \le p_\Ncal$ has been used in the last equality.

The binomial coefficient in \refeq{Cstarg} is at most $\binom{A}{2}$.
With \refeq{plusnormass}, this gives
\begin{align}
    \|\Cbf^{*}g\|_{\Phi}
&\le
    \binom{A}{2} \|\Cbf\|_{\Phi} \|g \|_{\Phi}
\end{align}
and hence
\begin{align}
    \frac{1}{2} \left| \pair{\Delta_{\Cbf} F,g}_{\phi}  \right|
&\le
    \|F\|_{T_{\phi}}
    \|\Cbf^{*}g\|_{\Phi}
\le
    \|F\|_{T_{\phi}}
    \binom{A}{2} \|\Cbf\|_{\Phi} \|g \|_{\Phi}
,
\end{align}
and \refeq{DCKzz} follows by taking the supremum over $g \in B (\Phi)$.
\end{proof}

\subsection{The main integration estimate}

In this section, we prove Proposition~\ref{prop:EK}.
For this, we adopt the conjugate fermion fields
setting described in Section~\ref{sec:cff},
with fields $\psi, \bar\psi$.
As noted below the statement of Proposition~\ref{prop:EK}, it suffices
to prove the bound \refeq{EKz}.
The proof is based on the following lemma,
which is known as \emph{Gram's inequality}.
A proof of Lemma~\ref{lem:Gramineq}
can be found in \cite[Lemma~1.33]{FKT02}.

\begin{lemma}
\label{lem:Gramineq}
Let $H$ be a Hilbert space with inner product $\langle \cdot, \cdot \rangle$.
If $u_i,v_i \in H$ for $i=1,\ldots, n$, then
\eq
    \left| \,\det \left( \langle u_i,v_j\rangle \right)_{1 \le i,j \le n} \right|
    \le
    \prod_{i=1}^n \langle u_i, u_i \rangle^{1/2}\langle v_i, v_i \rangle^{1/2}.
\en
\end{lemma}

Recall that $C_f$ can be interpreted as a test function as described
above the statement of Proposition~\ref{prop:EK}.
Let $E$ be the test function defined
by $E_z = \Ebold_{\Cbf_f} \psi^z$
for $z \in \vec\Lambdabold_f^*$,
with the convention $E_\varnothing =1$.  For $z \in \vec\Lambdabold^*
\setminus \vec\Lambdabold_f^*$ we set $E_z=0$.

\begin{lemma}
\label{lem:ell2}
If
$\|C_f\|_{\Phi} \le 1$ then $\|E\|_\Phi \le 1$.
\end{lemma}

\begin{proof}
To simplify the notation, we drop the subscript $f$ from $C_f$.
By \refeq{JF}, we may assume that $\psi^z$ has the form
$\psib_{x_1}\psi_{y_1}\cdots\psib_{x_p}\psi_{y_p}$, in which case
\eq
\label{e:Edet}
    E_z  = \det C_{x,y}.
\en
Let $\lambda_{\alpha,z} = (\prod_{i=1}^p\lambda_{\alpha_i',x_i})
( \prod_{i=1}^p\lambda_{\alpha_i'',y_i})$ with
$\lambda_{\alpha_i',x_i}= w_{\alpha_i',x_i}^{-1}\nabla^{\alpha_i'}$
and
$\lambda_{\alpha_i'',y_i}=
w_{\alpha_i'',y_i}^{-1}\nabla^{\alpha_i''}$,
with $\nabla^{\alpha_i'}$ acting on the $x_i$ variable and $\nabla^{\alpha_i''}$
acting on the $y_i$ variable.
It suffices to prove that
\eq
\label{e:ECbd}
    |\lambda_{\alpha,z} E_{z} |
    \leq
    \prod_{i=1}^p
    \left(\lambda_{\alpha_i',u}\lambda_{\alpha_i',v} C_{u,v}|_{u=v=x}
    \right)^{1/2}
    \left(\lambda_{\alpha_i'',u}\lambda_{\alpha_i'',v} C_{u,v}|_{u=v=y_i}
    \right)^{1/2},
\en
since \refeq{ECbd} implies the inequality
\eq
    \|E \|_{\Phi } \le \sup_{p \ge 1}
    \|C\|_{\Phi}^p .
\en

By \refeq{Edet} and the fact the determinant is linear in rows and columns,
\eq
\label{e:ECbd1}
    |\lambda_{\alpha,z} E_{z}|
    =
    |\lambda_{\alpha,z} \det C_{x,y}|
    =
     |\det(\lambda_{\alpha,z}  C_{x,y})|.
\en
We rewrite the determinant as follows.
Let $V$ be the vector space of all functions $f:\Lambda
\rightarrow \Cbold$.
Given functions $h,k\in V$, we define
\eq
    (h,k) = \sum_{x\in\Lambda} h_x  k_x.
\en
Then we define $f_i,g_i\in V$ by
\eq
    (\lambda_{\alpha_i',x_i } k)_{x_i}
    = ( \delta_{x_i}, \lambda_{\alpha_i',x_i } k)
    = (\lambda_{\alpha_i',x_i }^\dagger  \delta_{x_i}, k)
    = (f_i,k)
\en
and
\eq
    (\lambda_{\alpha_i'',y_i }h )_{y_i} = ( \lambda_{\alpha_i'',y_i } h, \delta_{y_i})
    = (h,\lambda_{\alpha_i'',y_i }^\dagger  \delta_{y_i})
    = (h,g_i).
\en
We define an inner product on $V$ by
\eq
    \pair{f,g}
=
    \sum_{x \in \Lambda} f_x C_{x,y}\bar{g}_y.
\en
By definition, for $i,j \in \{1,\ldots, p\}$,
\eq
    \lambda_{\alpha_i',x_i} \lambda_{\alpha_j'',y_j} C_{x_i,y_{j}}
    =
    \pair{f_{i},g_{j}}  ,
\en
and thus $\det (\lambda_{\alpha,z}  C_{x,y} ) = \det(\pair{f_{i},g_{j}})$.
By Lemma~\ref{lem:Gramineq},
\eq
     |\det (\lambda_{\alpha,z}  C_{x,y} ) |
=
     |\det\left( \pair{f_{i},g_{j}}\right) |
\le
    \prod_{i=1}^{n} \pair{f_{i}, f_i}^{1/2} \pair{g_i,g_i}^{1/2}.
\en
For the right-hand side, we use
\eq
    \pair{f_{i}, f_i}
=
    \pair{
    \lambda_{\alpha_i',x_i }^\dagger \delta_{x_i},
    \lambda_{\alpha_i',x_i }^\dagger  \delta_{x_i}
    }
=
    \lambda_{\alpha_i',u }
    \lambda_{\alpha_i',v }
    C_{u,v}\vert_{u=v=x_i}
    ,
\en
and similarly for $\pair{g_i,g_i}$.
With \refeq{ECbd1}, this proves \refeq{ECbd} and completes the proof.
\end{proof}

Proposition~\ref{prop:EK} is a consequence of the following lemma
(with $h=1$),
which establishes \refeq{EKz}.

\begin{lemma}
\label{lem:EKzz}
In the conjugate fermion field setting of
Section~\ref{sec:cff},
suppose that the covariance satisfies $\|C_f\|_{\Phi(w')}\le 1$.
If $F \in \Ncal (\Lambdabold \sqcup \Lambdabold')$
and $h: \R^{\Lambdabold'_{b}} \to \C$, then
\eq
\label{e:EKzzz}
    \| \Ex_{\Cbf} h F  \|_{T_\phi (w)}  \le
    \Ex_{\Cbf_b}
    \left[|h(\xi)|\,
    \|F  \|_{\Ttimes_{\phi\sqcup\xi} (w\sqcup w')}
    \right]
    .
\en
\end{lemma}

\begin{proof}
By definition, we can write $F = \sum_{z_f \in
(\overrightarrow{\Lambdabold_{f}\sqcup
\Lambdabold_{f}'})^{*}}\frac{1}{z_f!}  F_{z_f}\psi^{z_f}$, with
$F_{z_f}=F_{z_f}(\phi\sqcup\xi)$.
Given $z_{f}$, let $y$ be the
subsequence of $z_{f}$ such that $y \in \Lambdabold_{f}^{*}$, and let
$y'$ be the complementary subsequence of components of $z_{f}$ in
$\Lambdabold_{f}'$. The operator $\Ebold_{\Cbf}$ acts only on the
$\xi$ and $\psi'$ variables.  In particular,
\begin{align}
    \Ex_{\Cbf_f} \psi^{z_{f}}
&=
    \sgn (z_{f},y\concat y')
    \Ex_{\Cbf_f} \psi^{y\concat y'}
\nnb
&=
    \sgn (z_{f},y\concat y')
    \psi^{y}\Ex_{\Cbf_f} \psi^{y'}
=
    \sgn (z_{f},y\concat y')E_{y'}
    \psi^{y}
,
\end{align}
where $\sgn (z_{f},y\concat y')$ denotes the sign of
the permutation that maps $z_f$ to $y\concat y'$, and
$E$ denotes the test function of
Lemma~\ref{lem:ell2}.
Therefore, by \refeq{ECbf},
\begin{align}
    \Ex_{\Cbf} h F
&=
    \sum_{z_{f}\in (\overrightarrow{\Lambdabold_{f}\sqcup \Lambdabold_{f}'})^{*}}
    \frac{1}{z_{f}!}
    (\Ex_{\Cbf_b} h  F_{z_{f}}  )
    \sgn (z_{f},y\concat y')E_{y'}\psi^{y}
.
\end{align}
For $g \in \Phi (w)$, we define $E^{*}g \in \Phi (w\sqcup w')$ by
\begin{equation}
    (E^{*}g)_{z}
=
    \begin{cases}
    \sgn (z_{f},y\concat y')E_{y'}g_{z_{b}\concat y }
    &z_{b} \in \Lambdabold_{b}^{*}\\
    0
    &z_{b} \not \in \Lambdabold_{b}^{*}
    \end{cases}
\end{equation}
for $z \in (\overrightarrow{\Lambdabold \sqcup \Lambdabold'})^{*}$.
Then
\begin{align}
    \pair{\Ex_{\Cbf} h  F,g}_{\phi}
&=
    \sum_{z\in (\overrightarrow{\Lambdabold\sqcup \Lambdabold_{f}'})^{*}}
    \frac{1}{z!}
    (\Ex_{\Cbf_b} h  F_{z}  )
    (E^{*}g)_{z}
=
    \Ex_{\Cbf_b}
    [h(\xi) \pair{F, E^{*}g}_{\phi\sqcup\xi}]
,
\end{align}
and hence
\begin{align}
\label{e:EFgbd}
    \left| \pair{\Ex_{\Cbf_b} h  F,g}_{\phi}  \right|
&\le
    \Ex_{\Cbf_b} [|h(\xi)|\, |\pair{  F, E^{*}g}_{\phi\sqcup\xi}|]
\nnb &
\le
    \left(\Ex_{\Cbf_b} [|h(\xi)|\,\|F\|_{T_{\phi\sqcup\xi} (w\sqcup w')} ] \right)
    \|E^{*}g\|_{\Phi (w\sqcup w')}
.
\end{align}
Derivative operators $\nabla^\alpha$ do not act on the sgn function,
so we may apply \refeq{plusnormass-2} and then
Lemma~\ref{lem:ell2} to conclude that $\|E^{*}g\|_{\Phi
(w\sqcup w')} \le \|g\|_{\Phi (w)}$.  Then
\refeq{EKzzz} follows
by taking the supremum over $g \in B (\Phi(w))$ in \refeq{EFgbd},
and the proof is complete.
\end{proof}

Since $(E^{*}g)_z$ vanishes by definition whenever $z$ contains
an entry in $\Lambdabold_b'$, the above proof shows that
\refeq{EKzzz} could be strengthened by replacing the semi-norm on the
right-hand side by the smaller semi-norm which does not involve
derivatives with respect to the boson fluctuation field $\xi$.

\subsection{Expectation of the fluctuation-field regulator}
\label{sec:ffr}

The main result of this section is Lemma~\ref{lem:EG2zz}, which
immediately gives Proposition~\ref{prop:EG2}.
In preparation for Lemma~\ref{lem:EG2zz}, we prove three preliminary lemmas.
The first of these is proved in \cite[Lemma~6.28]{Bryd09}, and
a precursor of the second is \cite[Lemma~B.2]{BGM04}.

\begin{lemma}
\label{lem:integrability2}
Let $(\xi_{a})_{a \in \Acal}$
be a finite set of Gaussian
random variables with covariance $C$.  Suppose that the
largest eigenvalue of $C$ is less than $\frac 12$.
Let $(\xi,\xi) = \sum_{a \in \Acal} \xi_{a}^{2}$.  Then
\begin{equation}
    \Ebold \, e^{\frac{1}{2}(\xi,\xi)}
    \le
    e^{\sum_{a \in \Acal} C (a ,a )}.
\end{equation}
\end{lemma}

\begin{proof}
Let $t \in (0,1)$.
It suffices to show that
\begin{equation}
    \frac{d}{dt }
    \ln \Ebold \, e^{\frac{t }{2} (\xi,\xi)}
    \le
    \sum_{a \in \Acal} C (a ,a),
\end{equation}
since the desired inequality then follows by integration over $t \in (0,1)$.

Let $A$ be the inverse of the matrix $C$. The eigenvalues
of $A$ are at least $2$ by the hypothesis on $C$, so the inverse
matrix $C_{t}= (A-t)^{-1}$ exists.  Let $\Ebold_t$
denote the Gaussian expectation with covariance $C_t$. Then
\begin{align}
    \begin{split}
    \frac{d}{dt }
    \ln \Ebold \, e^{\frac{t }{2} (\xi,\xi)}
    &=
    \frac{1}{2}
    \Ebold_t (\xi,\xi)
    =
    \frac{1}{2}
    \sum_{a \in \Acal} C_{t} (a ,a)
    = \frac 12 {\rm Trace}\, C_t
    = \frac 12 \sum_\lambda (\lambda^{-1} - t)^{-1},
    \end{split}
\end{align}
where the sum over $\lambda$ runs over the eigenvalues of $C$
(with multiplicity).
Since each $\lambda$ is at most $\frac 12$ by hypothesis,
$(\lambda^{-1}-t)^{-1} = \lambda (1-t \lambda)^{-1}
\leq 2 \lambda$, and hence
\begin{equation}
    \frac{d}{dt }
    \ln \Ebold \, e^{\frac{t }{2} (\xi,\xi)}
    \le
    \sum_{\lambda}
    \lambda
    = \mathrm{Trace}\, C
    = \sum_{a \in \Acal} C (a ,a),
\end{equation}
which completes the proof.
\end{proof}

\begin{lemma}[Lattice Sobolev inequality]
\label{lem:sobolev2}
Let $f:B \rightarrow \C$, where $B\in \Bcal$ is a block
of side length $R$. Let $\nabla_R=R\nabla$.  Then for any $x \in B$,
\begin{equation}
    \label{e:sobolev2}
    |f (x)|^{2}
    \le
    2^{3d+2} R^{-d}
    \sum_{y \in B}
    \sum_{|\alpha|_\infty \le 1}
    | \nabla_R^{\alpha }f (y)|^{2}.
\end{equation}
\end{lemma}

\begin{proof}
We can choose coordinates on $B$ such that $B=\{0,1,\dotsc ,R-1 \}^{d}$.
Let $g:B\rightarrow \Rbold$ be any function that vanishes on
$\cup_{i=1}^{d}\{(x_{1},\dotsc ,x_{d}) \in B: x_{i} = 0 \}$. Then
we have the telescoping sum
\begin{equation}
    g (x)
    =
    \sum_{y:y_{i}< x_{i}\, \forall i} \nabla^{e_1}\dotsb \nabla^{e_d} \,g (y) .
\end{equation}
Therefore, by the Cauchy--Schwarz inequality,
\begin{equation}
    |g (x)|
    \le
    \sum_{y \in B} |\nabla^{e_1}\dotsb \nabla^{e_d} \,g (y)|
    \le
    \big(
    |B|\sum_{y \in B} |\nabla^{e_1}\dotsb \nabla^{e_d} \,g (y)|^{2}
    \big)^{1/2}.
\end{equation}
We apply this to $g (x) = x_{1}\dotsb x_{d}f (x)$, for points $x \in B$ with
each coordinate $x_{i}\ge R/2$.  This gives
\begin{equation}
    |f(x)|
    \le
    \left(\frac 2R \right)^{d}
    |x_{1}\dotsb x_{d}f(x)|
    \le
    2^{d}
    \big(
    |B|^{-1}
    \sum_{y \in B} |\nabla^{e_1}\dotsb \nabla^{e_d}y_{1}\dotsb y_{d}f (y)|^{2}
    \big)^{1/2}.
\end{equation}
We evaluate the derivatives using $ \nabla^{e_i}
y_{i} h (y) = y_{i}  \nabla^{e_i} h (y) + \nabla^{e_i} h (y)+ h (y)$.
Since $y_{i}  \le R$,
\begin{equation}
    |f(x)|^{2}
    \le
    2^{2d}
    |B|^{-1}\sum_{y \in B}
    \big(
    \sum_{\alpha \in \{0,1\}^{d}}
    2
    | \nabla_R^{\alpha} f (y)|
    \big)^{2}
    \le
    2^{2d+2}
    |B|^{-1}\sum_{y \in B}
    2^{d}\sum_{\alpha \in \{0,1\}^{d}}
    | \nabla_R^{\alpha} f (y)|^{2}.
\end{equation}
Since this holds for all functions $f$ we can change variables by
reflections through hyperplanes bisecting $B$ so as to remove the
assumption that every coordinate $x_i$ obeys $x_{i}\ge R/2$.  These reflections
turn forward derivatives into backward derivatives, and we obtain
\eqref{e:sobolev2} by noticing that the absolute value of a backward
derivative equals the absolute value of a forward derivative at a neighbouring
point.
\end{proof}

Recall the definition of $G(X,\phi)$ in Definition~\ref{def:ffregulator},
for $X\in \Pcal$ a polymer as in Definition~\ref{def:blocks}.

\begin{lemma}
\label{lem:Gtxi}
For $X  \subset \Lambda$, $t \ge 0$, and $\phi \in \C^{\Lambda}$,
\eq
    G^t(X,\phi)
    \le
    \exp\left[
    \frac{1}{2}
    \sum_{y \in X^{\Box}}
    \sum_{|\alpha|_{1} \le d+p_{\Phi}}
    | \xi (y,\alpha )|^{2}
    \right],
\en
where $\xi (y,\alpha) =
c t^{1/2} R^{-d/2} \ell^{-1} \nabla_R^{\alpha }\phi (y)$
for some constant $c$ depending only on $d$.
\end{lemma}

\begin{proof}
By definition,
\eq
\label{e:GXbfluct}
    G^t (X,\phi)
    =
    \exp \left[ t\sum_{x \in X}
    |B_{x}|^{-1}\|\phi\|_{\Phi (B_{x}^{\Box},\ell)}^2\right],
\en
so it suffices to show that
\begin{align}
    t \sum_{x \in X} |B_{x}|^{-1}\|\phi\|^2_{\Phi(B^\Box_{x},\ell)}
    &\le
    \frac{1}{2}
    \sum_{y \in X^{\Box}}
    \sum_{|\alpha|_{1} \le d+p_{\Phi}}
    | \xi (y,\alpha )|^{2}.
\end{align}
Throughout the proof, $c$ denotes a $d$-dependent constant whose value
may change from line to line.  Note that for $B\in \Bcal$,
$B^\Box$ is a cube (since connectivity of blocks can be via corners)
whose side length is a $d$-dependent multiple of $R$.  We first apply
Lemma~\ref{lem:sobolev2} with $f(x)=\nabla_R^\alpha \phi(x)$ and $B$
replaced by $B^\Box$ to obtain, for $x \in B^\Box$,
\begin{align}
    |\nabla_R^\alpha \phi(x)|^2
    &\le
    cR^{-d}
    \sum_{y \in B^\Box}
    \sum_{|\alpha'|_{\infty} \le 1}
    |  \nabla_R^{\alpha + \alpha' }\phi (y)|^{2}.
\end{align}
From this, we obtain
\begin{align}
    \|\phi\|_{\Phi(B^\Box,\ell)}^2
    &\le
    \max_{|\alpha|_{1} \le p_{\Phi} ,x \in B^\Box}
     |\ell^{-1} \nabla_R^{\alpha}\phi (x)|^{2}
     \le
    cR^{-d}
    \sum_{y \in B^\Box}
    \sum_{|\alpha|_{1} \le d+p_{\Phi}}
    | \ell^{-1}\nabla_R^{\alpha }\phi (y)|^{2}.
\end{align}
If $y \in B_{x}^{\Box}$ then $x \in B_{y}^{\Box}$ and
$|B_{y}^{\Box}|/|B|$ is bounded by a geometric constant.  With a
larger value of $c$, this gives
\begin{align}
    t \sum_{x \in X} |B_{x}|^{-1}\|\phi\|^2_{\Phi(B_{x}^\Box,\ell)}
    &\le
    ctR^{-d}
    \sum_{y \in X^\Box}
    \sum_{|\alpha|_{1} \le d+p_{\Phi}}
    | \ell^{-1}\nabla_R^{\alpha }\phi (y)|^{2}
    \nnb & =
    \frac{1}{2}
    \sum_{y \in X^{\Box}}
    \sum_{|\alpha|_{1} \le d+p_{\Phi}}
    | \xi (y,\alpha )|^{2},
\end{align}
and the proof is complete.
\end{proof}

Now we restate, and prove, Proposition~\ref{prop:EG2} as the following
lemma.  Recall that the $\Phi^+(\ell)$ norm is the
$\Phi(\ell)$ norm with $p_\Phi$ increased to $p_\Phi +d$.

\begin{lemma}
\label{lem:EG2zz} Let $t \ge 0$, $\Econstg >1$, and let $X
\subset \Lambda$.  There exists a (small) positive constant
$c (\Econstg)$, which is independent of $R$,
such that if $\|\Cbf_b\|_{\Phi^+(\ell)}\le
c (\Econstg )t^{-1}$, then
\begin{equation}
\label{e:EG2zz}
    0 \leq \Ex_{\Cbf_b}  G^t(X,\phi)  \le \Econstg^{R^{-d}|X|}.
\end{equation}
\end{lemma}

\begin{proof}
By Lemma~\ref{lem:Gtxi},
\begin{equation}
    \Ex_{\Cbf_b}  G^t(X,\phi)
    \le
    \Ex_{\Cbf_b}
    \exp\left[
    \frac{1}{2}
    \sum_{y \in X^{\Box}}
    \sum_{|\alpha|_{1} \le d+p_{\Phi}}
    | \xi (y,\alpha )|^{2}
    \right].
\end{equation}
The variables $\xi (x,\alpha)$ are Gaussian and we
denote their covariance by $Q$.
The largest eigenvalue $\lambda_{\rm max}$ of $Q$ is at most
the norm of $Q$ considered a convolution operator on $l^2(X^\Box)$.
Therefore, using Young's inequality we obtain
\eq
\label{e:YI}
    \lambda_{\rm max} \leq \sup_{f : \|f\|_2 \leq 1} \|Q*f\|_2
    \leq \|Q\|_1
    \leq
    cR^{d} \|Q\|_\infty.
\en
Since $Q$ is a positive-definite function, its maximum value occurs
on the diagonal, and obeys
\eq
    \|Q\|_\infty \leq ct R^{-d}
    \max_{|\alpha|_{1} \le d+p_{\Phi},\, x\in X^\Box}
    |\ell^{-2}\nabla_{R}^{2\alpha}\Cbf_{b;x,x}|
    \le
    c t R^{-d}  \|\Cbf_b\|_{\Phi^+ },
\en
so
\begin{equation}
    \lambda_{\rm max} \leq
    ct  \|\Cbf_b\|_{\Phi^+} .
\end{equation}
This will be less than $\frac 12$ if
$\|\Cbf_b\|_{\Phi^+} \le c (d)t^{-1}$ with
$c (d)$ sufficiently small.
We may therefore apply Lemma~\ref{lem:integrability2} with $\xi_{a}$ replaced by
$\xi (x,\alpha)$. This gives
\begin{equation}
    \Ex_{\Cbf_b}   G^t(X,\phi)
    \le
    e^{\sum_{y\in X^{\Box}}\sum_{|\alpha|_{1}
    \le d+p_{\Phi}} \text{Var} (\xi (y,\alpha))}.
\end{equation}
Since $\text{Var} (\xi (y,\alpha)) \le c t R^{-d} \|\Cbf_b\|_{\Phi^+
}$ this gives
\begin{equation}
    \Ex_{\Cbf_b}
    G^t(X,\phi)
    \le
    e^{ct \|\Cbf_b\|_{\Phi^+} R^{-d}|X^{\Box}|},
\end{equation}
and the desired result follows since $|X^{\Box}| \le a |X|$
for some $a=a(d)$.
\end{proof}

\section*{Acknowledgements}

The work of both authors was supported in part by NSERC of Canada.
DB gratefully acknowledges the support and hospitality of
the Institute for Advanced Study at Princeton and of Eurandom during part
of this work.
GS gratefully acknowledges the support and hospitality of
the Institut Henri Poincar\'e, and of the Kyoto
University Global COE Program in Mathematics,
during stays in Paris and Kyoto where part of this work was done.
We thank Beno\^it Laslier for many helpful comments,
and an anonymous referee for numerous pertinent suggestions.

\bibliography{../../bibdef/bib}
\bibliographystyle{plain}

\end{document}